\documentclass[sigconf,anonymous=false]{acmart}

\pdfoutput=1

\usepackage[autostyle]{csquotes}
\newtheorem{thm}{Theorem}[section]
\newtheorem{defn}{Definition}[section]
\DeclareMathOperator*{\E}{\mathbb{E}}
\DeclareMathOperator*{\Prob}{\mathbb{P}}
\allowdisplaybreaks
\usepackage[ruled,vlined,resetcount]{algorithm2e}
\RestyleAlgo{boxruled}
\usepackage[most]{tcolorbox}
\tcbset{
    center title,
    left=0pt,
    right=0pt,
    top=0pt,
    bottom=0pt,
    colback=blue!10!white,
    colframe=blue!50!black,
    width=\dimexpr\columnwidth\relax,
    enlarge left by=0mm,
    boxsep=5pt,
}
\graphicspath{{./Figures/}}
\usepackage{subcaption}

\usepackage{apxproof}

\begin{document}
\title{Improved Pan-Private Stream Density Estimation}
\titlenote{A preliminary version of this work was performed in \cite{digalakis2018thesis}.}

\author{Vassilis Digalakis Jr}
\authornote{Work performed while the author was at the Technical University of Crete.}
\affiliation{
Massachusetts Institute of Technology}
\email{vvdig@mit.edu}
\author{George N. Karystinos}
\affiliation{
Technical University of Crete}
\email{karystinos@telecom.tuc.gr}
\author{Minos N. Garofalakis}
\affiliation{Athena Research Center \\ 
Technical University of Crete}
\email{minos@athenarc.gr}

\begin{abstract}
Differential privacy is a rigorous definition for privacy that guarantees that any analysis performed on a sensitive dataset leaks no information about the individuals whose data are contained therein.
In this work, we develop new differentially private algorithms to analyze streaming data.
Specifically, we consider the problem of estimating the density of a stream of users (or, more generally, elements),
which expresses the fraction of all users that actually appear in the stream.
We focus on one of the strongest privacy guarantees for the streaming model, namely user-level pan-privacy, which ensures that the privacy of any user is protected,
even against an adversary that observes, on rare occasions, the internal state of the algorithm.
Our proposed algorithms employ optimally all the allocated \enquote{privacy budget}, are specially tailored for the streaming model,
and, hence, outperform both theoretically and experimentally the conventional sampling-based approach.
\end{abstract}



\keywords{Differential privacy; pan-privacy; streaming algorithms}

\maketitle

\section{Introduction}
The importance of privacy in the era of Big Data is well-understood.
Introduced in 2006 by Cynthia Dwork, Frank McSherry, Kobi Nissim, and Adam Smith \cite{dwork2006calibrating}, 
\emph{differential privacy} is a strong and mathematically rigorous guarantee,
that describes a promise, made by a data holder, 
to an individual that contributes its data to a dataset, 
that the individual's privacy will be protected.
As Dwork and Roth \cite{dwork2014algorithmic} argue, differential privacy addresses the paradox of learning nothing about an individual while learning useful information about a population.
In addition to being a strong privacy guarantee, differential privacy as a definition has mathematical properties
that facilitate the design of algorithms which satisfy it.
As a result, numerous methods that realize the differential privacy guarantee have been developed in the past few years,
and now we are definitely one step closer to privacy-preserving data analysis.

At the same time, the paradigm of our data living in a static database that is protected by a curator (e.g. database administrator) belongs to the past to a large extent.
Nowadays, our data are everywhere and in various forms. For instance, they may be dynamically created and arrive continuously in a stream \cite{muthukrishnan2005data,garofalakis2016data}. 
Being able to monitor such a stream and extract statistics is important for many disciplines, such as epidemiology and real-time health monitoring.
We mention those particular applications due to their obvious connections with privacy; medical data are by definition sensitive.

\paragraph*{Outline \& Contributions:}
In this paper, we give a detailed and unified presentation of differential privacy for static and streaming datasets, and emphasize some points that often cause confusion in the literature.
We develop differentially private algorithms to analyze streaming data and address the particular problem of estimating the density of a stream
of users.
We offer one of the strongest privacy guarantees for the streaming model,
namely user-level pan-privacy,
which ensures that the privacy of any user is protected,
even against an adversary that observes, on rare occasions,
the internal state of the algorithm.

The key contributions of the paper are the following. 
We formally describe differential privacy in the streaming model and analyze in-depth the existing definitions and approaches. 
We provide for the first time a detailed analysis of the sampling-based, pan-private density estimator, proposed by Dwork et al. \cite{dwork2010pan} and
identify its main limitation, in that it does not use all the allocated privacy budget.
We examine a novel approach to modifying the original estimator,
based on optimally tuning the Bernoulli distributions it uses, and
analyze the theoretic guarantees that our modified estimator offers.
Next, we further modify the estimator by replacing the static sampling step it performs with an adaptive sampling approach that is specially tailored for the particular problem we consider (namely, Distinct Sampling);
our proposed estimator, which we call Pan-Private Distinct Sampling, significantly outperforms the original one.
Finally, we experimentally compare our algorithms.

\section{System Models \& Definitions}
In this section, we introduce the models we consider, namely the \emph{static} and the \emph{streaming} model, as well as the associated privacy definitions and mechanisms.
Before we begin, we make a few notational remarks.
By $\log$ we refer to the natural logarithm, whereas by $\log_2$ we refer to the logarithm with base $2$.
We slightly abuse notation and use, for example, $\text{Bernoulli}(p)$ to refer to a Bernoulli random variable with parameter $p$ (instead of $X \sim \text{Bernoulli}(p)$).
Similarly, we use $\text{Laplace}(\mu,b)$ to refer to a random variable that follows the Laplace Distribution with mean $\mu$ and variance $2b^2$, which is a symmetric, double-sided version of the exponential distribution and is defined as follows.
\begin{defn}[Laplace distribution] \label{defn_laplace}
A random variable $X \sim \text{Laplace}(\mu,b)$ ($b>0$) has probability density function:
$$f_X(x|\mu,b) = \frac{1}{2b} e^{-\frac{|x-\mu|}{b}} \ , \ x \in \mathbb{R}.$$
\end{defn}

\subsection{Differential Privacy for Static Datasets}
In this section, we formalize the notion of differential privacy for static datasets
and describe the most common mechanisms that are used as building blocks to provide the differential privacy guarantee.

\subsubsection{The Static Model}
We consider a model of computation, where a database or, more generally, a dataset $D$, i.e., a finite and, without loss of generality (w.l.o.g.), ordered collection of $T$ records/tuples/data points (the terms are used interchangeably) from a finite universe $\mathcal{D} = \mathcal{U} \times \mathcal{V}$,
contains the sensitive data $\mathbf{v}=(v_1,...,v_T)$ of data owners $\mathbf{u}=(u_1,...,u_T)$.
Each record in $D$  is a pair $D_t=(u_t,v_t)$ and it is assumed that both user ids and values belong in arbitrarily large but finite alphabets,
i.e., $\forall t$, $u_t \in \mathcal{U}$ and $v_t \in \mathcal{V}$, $|\mathcal{U}|,|\mathcal{V}|<\infty$.
Therefore, it is possible that multiple records correspond to the same user,
although, in practice, usually a single record corresponds to each user.
We also introduce the function $\text{hist}(D): \mathcal{D}^T \rightarrow \mathbb{N}^{|\mathcal{U} \times \mathcal{V}|}$
which is called the histogram representation of $D$ and
computes the (unnormalized) type (or frequency distribution) \cite{cover2012elements} of $D$, that is, the number of occurrences in $D$ of each record from $\mathcal{D}$.
The notation we just described is summarized in Table \ref{table notation 1}.

\begin{table}
\renewcommand{\arraystretch}{1.3}
\setlength{\tabcolsep}{2pt}
\centering
\begin{tabular}{|c|c|}
\hline
\textbf{Notation} & \textbf{Description}\\
\hline
\hline
$u_t \in \mathcal{U} = \{1,...,U\}$ & Universe of users/keys\\
\hline
$v_t \in \mathcal{V} = \{1,...,V\}$ & Universe of values\\
\hline
$D_t = (u_t,v_t) \in \mathcal{D} $ & $t$-th record \\
\hline
\begin{tabular}{c}
$D = (D_1,...,D_T) = ( \mathbf{u}_{T \times 1} , \mathbf{v}_{T \times 1} ), $ \\
$D \in \mathcal{D}^T = (\mathcal{U} \times \mathcal{V})^T$ 
\end{tabular}
& Dataset\\
\hline
$\textbf{hist}(D) : \mathcal{D}^T \rightarrow \mathbb{N}^{|\mathcal{U} \times \mathcal{V}|}$  & Histogram of dataset \\
\hline
\end{tabular}
\vspace{1em}
\caption{Notations (Part I).} \label{table notation 1}
\end{table}

The notion of adjacency between datasets informally refers to a pair of datasets that differ on a single record;
depending on the interpretation of the word \enquote{differ,} two definitions have been used in the literature.

\begin{defn}[Adjacency] \label{defn_adj}
Datasets $D$ and $D'$ are adjacent, denoted $adj(D,D')$, iff
\begin{itemize}
\item[\textbf{A:}] $D$ can be obtained from $D'$ by either adding or removing a single record, so
$$\|\textbf{hist}(D)-\textbf{hist}(D')\|_1 = \sum_{i=1}^{|{\mathcal{U} \times \mathcal{V}}|} |\textbf{hist}(D)_i-\textbf{hist}(D')_i| = 1,$$
\item[or] 
\item[\textbf{B:}] $D$ can be obtained from $D'$ by changing the value of a single record, so
$$|D|=|D'| \ \text{ and } \ 
\|\textbf{hist}(D)-\textbf{hist}(D')\|_1 = 2. $$
\end{itemize}
\end{defn}

Our ultimate goal is to protect the privacy of the data owners, whose data are contained in $D$.
We assume the existence of a trusted entity by the data owners, called the data holder or curator.
The data holder has direct access to the sensitive dataset and analyzes it,
ensuring that any output produced by the analysis does not violate the owners' privacy.
%
We also note that we focus on the non-interactive case, so that the analysis, i.e., the entire set of queries that is to be performed on the sensitive data, is decided and known in advance.

\subsubsection{Privacy Definition: Differential Privacy}
We are now ready to formally define differential privacy, which intuitively guarantees that a randomized algorithm accessing a sensitive dataset produces similar outputs on similar (adjacent) inputs.
As a result, the impact of any single record (individual) on the algorithm's output is negligible and, hence, no information is leaked about the individuals whose data are in the dataset.
\begin{defn}[Differential Privacy] \label{defn_dp}
A randomized algorithm $\textbf{Alg}: \mathcal{D} \rightarrow \mathcal{O}$ is $\varepsilon$-differentially private
if for all $O \subseteq \mathcal{O}$, and for all pairs of adjacent datasets $D,D'$,
$$ \Prob[\textbf{Alg}(D) \in O ] \ \leq \ e^{\varepsilon} \Prob[\textbf{Alg}(D') \in O ]$$
where the probability space is over the coin flips of $\textbf{Alg}$.
\end{defn}
The parameter $\varepsilon$, called the \emph{privacy budget}, quantifies the privacy risk;
smaller values of $\varepsilon$ imply higher privacy, as the distributions of outputs of the algorithm for adjacent inputs tend to \enquote{come closer.}
Depending on the approach we follow to defining adjacency, we end up with a (slightly) different flavor of differential privacy (Kifer and Machanavajjhala \cite{kifer2011no}).
If Definition \ref{defn_adj}(\textbf{A}) is used, we consider \emph{unbounded} differential privacy, whereas, if Definition \ref{defn_adj}(\textbf{B}) is adopted, we consider \emph{bounded} differential privacy.
Bounded differential privacy derives its name from the fact that the adjacent datasets involved have essentially the same size;
in unbounded differential privacy there is no such restriction.
In the literature, both approaches have been considered.

\subsubsection{Privacy Mechanisms}
Differential privacy is a definition, not an algorithm.
In practice, we are interested in developing algorithms that satisfy Definition \ref{defn_dp}, i.e., \emph{privacy mechanisms}, and hence offer the differential privacy guarantee to their input datasets.
As the reader may have noticed, a privacy mechanism is \emph{essentially} a \emph{randomized algorithm}, i.e.,
an algorithm that employs a degree of randomness as part of its logic and produces an output that is a random variable (or vector).
We next present two primitive differentially private mechanisms, which we use throughout our work,
namely Randomized Response and the Laplace Mechanism.
We note that both more general mechanisms (e.g., Exponential Mechanism) and task-specific mechanisms (e.g., Sparse Vector Technique, Multiplicative Weights Mechanism, Subsample and Aggregate Framework) have been developed;
we refer the interested reader to the monograph by Dwork and Roth \cite{dwork2014algorithmic} for a detailed presentation of such mechanisms.

\paragraph*{Randomized Response.}
Randomized response is a research method proposed by Warner \cite{warner1965randomized} that allows respondents to a survey on a sensitive issue to protect their privacy against the interviewer, while still providing credible answers.
Suppose that the input dataset consists of a binary record (a bit) per individual,
which indicates whether the individual does or does not have a particular property;
formally, let $u_t = t$, i.e., $\mathbf{u} = (1,...,U)$, and $\mathcal{V} = \{0,1\}$.
Assume that we wish to output $\mathbf{f}(D) = \mathbf{v} = (v_1,...,v_U)$.
Then, using Randomized Response, we instead output \cite{dwork2014algorithmic}
\begin{equation*}
\label{eqn rand resp} 
\mathbf{\tilde{f}}(D) = \mathbf{v} \odot \mathbf{b_1} + \mathbf{b_2} \odot (\mathbf{1} - \mathbf{b_1})
\end{equation*}
where $\mathbf{b_1}$ and $\mathbf{b_2}$ are $U$-dimensional random vectors of i.i.d. Bernoulli$(\frac{1}{2})$ entries and $\odot$ denotes the bitwise logical conjunction (AND) operator.

%
%
%
%
%
%
%
%

The power of randomized response is that it provides plausible deniability and directly perturbs the sensitive dataset (privacy by process).
As a result, even if an individual's record indicates that it has the property in question, the individual may still credibly argue that it does not.
Theorem \ref{thm_rand_resp} examines the privacy guarantees of Randomized Response.
\begin{thm}[Randomized response] \label{thm_rand_resp}
The version of randomized response described in Section \ref{eqn rand resp} satisfies $\log 3$-differential privacy.
\end{thm}

\paragraph*{Laplace Mechanism.} 
The Laplace Mechanism \cite{dwork2006calibrating} provides a differentially private mechanism for a numeric function $\mathbf{f}$ that has input a dataset and output, in general, a vector.
Specifically, given a function $\mathbf{f}: \mathcal{D}^T \rightarrow \mathbb{R}^N$ and an input $D \in \mathcal{D}^T$, the Laplace Mechanism outputs
$$\mathbf{\tilde{f}}(D) = \mathbf{f}(D) + (X_1, ... ,X_N) $$
where $X_i$ are i.i.d. $\text{Laplace}(0,\frac{\Delta \mathbf{f}}{\varepsilon})$ random variables and
$$\Delta \mathbf{f} = \max_{adj(D,D')} \|\mathbf{f}(D)-\mathbf{f}(D')\|_1$$
is called the $\ell_1$-sensitivity of $\mathbf{f}$.
Intuitively, the $\ell_1$-sensitivity captures the effect of a single record on the output of $\mathbf{f}$. Theorem \ref{defn_laplace_mech} examines the privacy guarantees of the Laplace Mechanism.
\begin{thm}[Laplace Mechanism] \label{defn_laplace_mech}
The Laplace Mechanism satisfies $\varepsilon$-differential privacy.
\end{thm}

\subsection{Differential Privacy for Streaming Datasets}
In this section, we discuss how the notion of differential privacy can be adjusted to the \emph{data stream model} (\cite{muthukrishnan2005data}, \cite{garofalakis2016data}, \cite{cormode2011synopses}).

\subsubsection{The Streaming Model}
We consider the cashier-register data stream (or streaming) model with unit updates,
where the input is represented as a finite sequence $S$ of $T$ updates (i.e., records ordered by time) from a finite universe $\mathcal{S}$ and, in particular,
each update is of the form $S_t = (u_t,v_t)$ with $v_t=1$, $\forall t$.  
Therefore, the stream can be viewed as a sequence of users and the function $\textbf{hist}(S)$ maps each user to a state that accumulates the user's updates, i.e., the number of times the user appears in the stream.
We remark that our work naturally generalizes to the cashier-register model where $v_t \in \mathbb{N}_{>0}$ and, 
although it does not apply to the (most general) turnstile streaming model where $v_t \in \mathbb{Z}$,
it can be extended to the special case where $v_t \in \{-1,+1\}$
under the additional constraint that $\textbf{hist}(S^t) \in \{0,1\}^t$.
This last model is of special interest in graph applications, where edges can be inserted or deleted and, at any time step, an edge can be either present or absent.
Finally, the particular problem we consider is that of estimating the density $d(S)$ of a given input stream, that is, the fraction of $\mathcal{U}$ that actually appears in $S$;
we need to do this in a single pass, in real-time, and using small space (much smaller than $|\mathcal{U}| \text{ or } T$).
A common approach is the use of approximation algorithms that offer $(\alpha,\beta)$-approximate answers, i.e., with probability at least $1-\beta$ the computed answer is within an additive/multiplicative factor $\alpha$ of the actual answer.
The notation we just described in summarized in Table \ref{table notation 2}.
\begin{table}
\renewcommand{\arraystretch}{1.3}
\setlength{\tabcolsep}{2pt}
\centering
\begin{tabular}{|c|c|}
\hline
\textbf{Notation} & \textbf{Description}\\
\hline
\hline
$S_t = (u_t,1) \in \mathcal{S} = \mathcal{U}$ & Stream update at time $t$\\
\hline
$S^t = (S_1,...,S_t)$ & Stream after $t$ updates\\
\hline
$S=S^T \in \mathcal{S}^T = \mathcal{U}^T$ & Final input stream \\
\hline
$\textbf{hist}(S^T)_u = \sum_{i=1}^T{ \mathbf{1}(u_i=u) }$ & State of user $u \in\mathcal{U} $ \\
\hline
$d(S^T) = \frac{ \|\textbf{hist}(S^T)\|_0 }{ U }$ & Density of stream \\
\hline
\end{tabular}
\vspace{1em}
\caption{Notations (Part II).}
\label{table notation 2}
\end{table}

We next examine the notion of adjacency in the streaming model and introduce the following two approaches.
\begin{defn}[Event-level Adjacency]\label{def_adj_1-event}
Data streams $S$ and $S'$ are event-level adjacent, if they differ in a single update $S_t$ for some $t \in \{1,...,T\}$.
\end{defn}
\begin{defn}[User-level Adjacency]\label{def_adj_1}
Data streams $S$ and $S'$ are user-level adjacent, denoted as $adj(S,S')$, if they differ in all (or some) updates that correspond to a single user $u \in \mathcal{U}$.
\end{defn}
Depending on which definition is used, we get \emph{event-level privacy}, which guarantees that an adversary cannot distinguish whether update $S_t$ did or did not occur,
and \emph{user-level privacy}, which guarantees that an adversary cannot distinguish whether $u$ did or did not ever appear, independently of the actual number of appearances of $u$.
The privacy level affects the amount of perturbation used, so a much stronger guarantee like user-level privacy requires excessively more perturbation.
Kellaris et al. \cite{kellaris2014differentially} attempt to bridge the gap between event-level and user-level privacy and develop the $w$-event privacy $w$-event privacy framework, which protects all appearances of any user, occurring within a window of $w$ time steps.

As in the case of static datasets, the word \enquote{differ} in Definitions \ref{def_adj_1-event} and \ref{def_adj_1} can be interpreted in two ways.
According to the first interpretation (used, e.g., in \cite{dwork2010pan}), 
stream $S$ can be obtained from $S'$ by adding/removing an update (in Definition \ref{def_adj_1-event}) or adding/removing all (or some) updates that refer to a user (in Definition \ref{def_adj_1}), following Definition \ref{defn_adj}(\textbf{A}).
As a result, the two streams do not have the same length.
According to the second interpretation (used, e.g., in \cite{mir2011pan}),
stream $S$ can be obtained from $S'$ by changing the value of an update (in Definition \ref{def_adj_1-event}) or replacing all (or some) updates that refer to a user with updates that refer to another user (in Definition \ref{def_adj_1}), following Definition \ref{defn_adj}(\textbf{B}).
As a result, the two streams have the same length.


\subsubsection{Privacy Definition: Pan-Privacy}
Two privacy definitions have been proposed for the streaming model \cite{dwork2010differential}, namely pan-privacy \cite{dwork2010pan} and differential privacy under continual observation \cite{dwork2010continual}.
We introduce them through the following example.

\begin{example}
Consider an algorithm that takes as input a data stream $S=(S_1,...,S_T)$ and, upon arrival of each element $S_i$, computes some estimate $f_i = f(S_1,...,S_i)$.
If we were in the static model, so that any computation would be performed after the entire stream was processed, 
then, according to Definition \ref{defn_dp},
differential privacy would ensure that privacy is preserved against an adversary that, at time $T$, observed the output $f_T$.
\emph{Differential privacy under continual observation} ensures that privacy is preserved against an adversary that observes the entire sequence of outputs $f_1,...,f_T$.
\emph{Pan-privacy}, which we formally define in Definition \ref{def_pan_privacy}, ensures that privacy is preserved against an adversary that, 
at time $j \in \{1,...,T\}$, observes the internal state of the algorithm (and hence can compute the corresponding output $f_j$)
and, at time $T$, observes the final output $f_T$.
In other words, pan-privacy enforces an additional differential privacy constraint to the internal state of the algorithm.
The two definitions can be combined, by ensuring that all the internal state, the entire output sequence, and their joint distribution satisfy differential privacy. 
\end{example}

We proceed by formalizing the notion of pan-privacy, which aims to protect against an adversary that can, on rare occasions, observe the internal state of the algorithm.
This is, of course, in addition to the standard -for differential privacy- assumptions about the adversary having arbitrary control over the input, arbitrary prior knowledge, and arbitrary computational power.

\begin{defn}[Pan-privacy]\label{def_pan_privacy}
Let $\textbf{Alg}$ be a randomized algorithm.
Let $\mathcal{I}$ denote its set of internal states and $\mathcal{O}$ denote its set of possible outputs.
Then $\textbf{Alg}$, mapping data streams of finite length $T$ to the range $\mathcal{I} \times \mathcal{O}$, is $\varepsilon$-pan-private against a single intrusion,
if, for all sets $I \subseteq \mathcal{I}$ and $O \subseteq \mathcal{O}$ and all pairs of adjacent data streams $S,S'$,
$$ \Prob[\textbf{Alg}(S) \in (I,O) ] \ \leq \ e^{\varepsilon} \Prob[\textbf{Alg}(S') \in (I,O) ]$$
where the probability space is over the coin flips of $\textbf{Alg}$.
\end{defn}
Definition \ref{def_pan_privacy} considers only a single intrusion. 
To handle multiple intrusions, we must consider interleavings of observations of internal states and output sequences.
We also have to differentiate between announced and unannounced intrusions.
In the former case (e.g., subpoena), the algorithm knows that an intrusion occurred, so it can re-randomize its state and handle multiple announced intrusions.
In the latter case (e.g., hacking), the algorithm can only tolerate a single unannounced intrusion and strong negative results have been proved for even two unannounced intrusions.

We remark that ordinary streaming algorithms based on sampling and sketching techniques do not provide the pan-privacy guarantee.
Sampling techniques maintain information about a subset of the users, so an intruder with access to the sample (the internal state of the algorithm) would violate the privacy of the sampled users.
Sketching techniques which are based on hashing, like the FM Sketch \cite{flajolet1985probabilistic} and the Count-Min Sketch \cite{cormode2005improved}, also cannot protect the privacy of the users against an adversary who has access to the hash functions used; the hash functions are part of the internal state.

\subsubsection{Related Work: Pan-Private Algorithms}
So far, two works have addressed the challenge of developing pan-private streaming algorithms.
They both examine variants of the same problems, applying different techniques.
Dwork et al. \cite{dwork2010pan} work in the cashier-register streaming model and develop algorithms based on sampling and randomized response.
Mir et al. \cite{mir2011pan} work in the turnstile streaming model and rely on both existing and novel sketches; they also develop a general noise-calibrating technique for sketches. 
The problems examined are the following:
\begin{itemize}
\item[-] Density estimation and distinct count: the fraction of $\mathcal{U}$ that appears and the number of users with nonzero state, respectively.

\item[-] $t$-cropped mean and $t$-cropped first moment: the average and sum, respectively, over all users, of the minimum of $t$ and the number of appearances of the user.

\item[-] Fraction of $k$-heavy hitters and $k$-heavy hitters count: the fraction and number, respectively, of users that appear at least $k$ times.

\item[-] $t$-incidence estimation: the fraction of users that appear exactly $t$ times.
\end{itemize}

\emph{In our work, we aim to offer the user-level pan-privacy guarantee (against a single intrusion).}
Our objective is to estimate the density of the given input stream
and our algorithms are based on the density estimator of Dwork et al. \cite{dwork2010pan}.
The algorithm proposed by Mir et al. \cite{mir2011pan} for the (similar) distinct count problem relies on the $\ell_0$ Sketch, which is due to Cormode et al. \cite{cormode2002comparing} and utilizes a family of distributions called $p$-stable (Indyk \cite{indyk2000stable}).
Despite being particularly interesting theoretically, the pan-private algorithm of Mir et al. \cite{mir2011pan} is extremely impractical as it involves an application of the exponential mechanism, which requires sampling twice from a space of $2^{b_S}$ sketches (where $b_S$ is the bit size of the sketch).
To make matters worse, in order to evaluate the exponential mechanism's scoring function, an $\ell_0$ norm minimization problem is solved for every single sketch, which translates to solving $2^{b_S}$ problems.

Finally, although our focus is on algorithms that produce a single output, we remark that Dwork \cite{dwork2010differential} shows how to modify the pan-private density estimator developed by Dwork et al. \cite{dwork2010pan} (which, we repeat, is the baseline for our work) to produce output continually.
The resulting continual observation density estimator guarantees \emph{user-level pan-privacy under continual observation}.
A similar technique can be applied to all the algorithms we develop to allow them to produce continual output while still preserving privacy.

\section{Conventional Pan-Private Density Estimator}
Dwork et al. \cite{dwork2010pan} proposed the first user-level pan-private algorithm for the density estimation problem (Algorithm \ref{DworkDE}). In this section, we provide a detailed presentation and analysis of that algorithm, which we call \verb|Dwork|.
Algorithm \ref{DworkDE} is randomized;
it takes as input the data stream $S$ whose density we wish to estimate, as well as the desired privacy budget and accuracy parameters.
Although the length $T$ of the stream is assumed known in advance,
as mentioned in the previous section, this assumption does not affect the analysis of the algorithm and can easily be dropped, so that the algorithm outputs the stream density when it receives a special signal.

To keep its state small, the algorithm maintains information only about a subset of the users, which is selected \emph{randomly} at the beginning.
The state of the algorithm is a bitarray with one \emph{random} bit per sampled user, drawn as described below.
If a user has not appeared in $S$, its entry is drawn from a uniform Bernoulli distribution, otherwise it is drawn from a slightly biased (towards $1$) Bernoulli distribution, no matter how many times the user has appeared.
The two distributions should be close enough to guarantee that the state satisfies differential privacy, but far enough to allow collection of aggregate statistics about the fraction of users that appear at least once.
When the algorithm outputs the final density estimate, it uses the \emph{random} Laplace Mechanism, which guarantees that the output also satisfies differential privacy.
Hence, the algorithm uses three degrees of randomness; namely, random selection of the subset of users, random generation of the bitarray, and random variation of the final density estimate.

\begin{algorithm}

\caption{\texttt{Dwork}\label{DworkDE}}

\DontPrintSemicolon

\KwIn{Data stream $S$, Privacy budget $\varepsilon$, Accuracy parameters $(\alpha,\beta)$}

\KwOut{Density $\tilde{d}(S)$}

Pick $m = \mathcal{O}(\frac{1}{\varepsilon^2 \alpha^2}\log{\frac{1}{\beta}})$
\ \ \textit{(or compute $m^*$ as described in subsection \ref{subSEC_sample_size} and set $m=m^*$)}\;

Sample a random subset $\mathcal{M} \subseteq \mathcal{U}$ of $m$ users (without replacement) and define an arbitrary ordering over $\mathcal{M}$\;

Create a bitarray $\mathbf{b}=[b_1\ ...\ b_{m}]$ and map $\mathcal{M}[i] \rightarrow b_i,\ \forall \ i \in \{1,...,m\}$\;

Initialize $\mathbf{b}$ randomly: $b_i \sim \text{Bernoulli}(\frac{1}{2}), \ \forall \ i \in \{1,...,m\}$\;

\For{t = 1 \KwTo T}{

	\If{$S_t \in \mathcal{M}$}{
		
		Find $i:\ \mathcal{M}[i]=S_t$\;
		
		Re-sample: $b_i \sim \text{Bernoulli}(\frac{1}{2}+\frac{\varepsilon}{4})$\;
			
	}
}

Return $\tilde{d}(S)\ = \ \frac{4}{\varepsilon}(\frac{1}{m}\sum_{i=1}^{m}{b_i}-\frac{1}{2}) \ + \ \text{Laplace}(0,\frac{1}{\varepsilon m})$\;

\end{algorithm}

We now make two important remarks concerning potential extensions of Algorithm \ref{DworkDE}.
The techniques described in the remarks apply (slightly modified) to all the algorithms we present, so we do not revisit them in our work.
\begin{itemize}
\item[-] Algorithm \ref{DworkDE} can tolerate a single (announced or unannounced) intrusion. Dwork et al. \cite{dwork2010pan} show how to handle multiple announced intrusions, by re-randomizing the state after each intrusion has occurred.
\item[-] Algorithm \ref{DworkDE} works in the cashier-register streaming model, that is, once a user appears in the stream, it cannot be deleted.
However, as we mentioned earlier, our work also applies to the case where a user $u$ may be both inserted and deleted (later on) from the stream.
In particular, if an update of the form $(u,insert)$ arrives, $u$'s bit is drawn from $\text{Bernoulli}(\frac{1}{2}+\frac{\varepsilon}{4})$, whereas if an update of the form $(u,delete)$ arrives, $u$'s bit is re-drawn from $\text{Bernoulli}(\frac{1}{2})$.
This allows Algorithm \ref{DworkDE} to perform pan-private graph density estimation as well.
\end{itemize}

\subsection{Analysis}
We first examine the privacy guarantees of Algorithm \ref{DworkDE}. We have the following theorem. For simplicity in presentation, we defer all proofs to the Appendix.

\begin{thm}\label{thm_DworkDE_priv}
If $\varepsilon \leq \frac{1}{2}$, then Algorithm \ref{DworkDE} satisfies $2\varepsilon$-pan-privacy.
\end{thm}

\begin{toappendix}
\begin{proof}[Proof of Theorem~\ref{thm_DworkDE_priv}]
Let $S,S'$ be two adjacent streams that differ on all occurrences of user $u \in \mathcal{U}$. Assume w.l.o.g. that $u \in S$ and $u \not\in S'$.

The \emph{state} of Algorithm \ref{DworkDE} satisfies $\varepsilon$-differential privacy. All the information that Algorithm \ref{DworkDE} stores as its state is the bitarray $\mathbf{b}$. We distinguish the following two cases.\\
i) $u \not\in \mathcal{M}$: Perfect privacy is guaranteed, as no information is stored on user $u$.\\
ii) $u \in \mathcal{M}$: W.l.o.g., let $b_1$ be the entry that corresponds to $u$ in the bitarray. Then, assuming that $u$ has already appeared in the stream when the adversary views $\mathbf{b}$, $b_1(S)$ is drawn from $\text{Bernoulli}(\frac{1}{2}+\frac{\varepsilon}{4})$ and $b_1(S')$ is drawn from $\text{Bernoulli}(\frac{1}{2})$. We thus have to bound the following probability ratios according to the differential privacy definition.
\begin{equation*}
\begin{split}
\frac{\Prob\left(\mathbf{b(S)}=[1 \ \underline{b_2} \ ...\ \underline{b_{m}}]\right)}{\Prob\left(\mathbf{b(S')}=[1 \ \underline{b_2} \ ...\ \underline{b_{m}}]\right)}
&=
\frac{\Prob\left(b_1(S)=1\right)}{\Prob\left(b_1(S')=1\right)} \\ 
&= 
\frac{\frac{1}{2}+\frac{\varepsilon}{4}}{\frac{1}{2}}
\ = \ 1+\frac{\varepsilon}{2} , \\
\frac{\Prob\left(\mathbf{b(S)}=[0 \ \underline{b_2} \ ...\ \underline{b_{m}}]\right)}{\Prob\left(\mathbf{b(S')}=[0 \ \underline{b_2} \ ...\ \underline{b_{m}}]\right)} 
&= 
\frac{\Prob\left(b_1(S)=0\right)}{\Prob\left(b_1(S')=0\right)} \\
&=
\frac{\frac{1}{2}-\frac{\varepsilon}{4}}{\frac{1}{2}} 
\ = 
\ 1-\frac{\varepsilon}{2} .
\end{split}
\end{equation*}

Since $ e^{-\varepsilon} \leq 1+\frac{\varepsilon}{2} \leq \sum_{k=0}^{\infty}{{\varepsilon^k}{k!}} = e^{\varepsilon} $ and  $ e^{-\varepsilon} \leq 1-\frac{\varepsilon}{2} \leq e^{\varepsilon} , \ \forall \ \varepsilon \in [0,\frac{1}{2}]$, it is implied that user $u$ is guaranteed $\varepsilon$-differential privacy against an adversary that observes $u$'s entry in $\mathbf{b}$.

The \emph{output} of Algorithm \ref{DworkDE} (conditioned on the state) also satisfies $\varepsilon$-differential privacy, as it is computed by independently applying the Laplace mechanism.
Specifically, since the sensitivity of $\tilde{d}$ is
$$\Delta \tilde{d} = \max_{\{S,S'\}:\ adj(S,S')}{|\tilde{d}(S) - \tilde{d}(S')|} = \frac{1}{m} ,$$
we add noise $\sim \text{Laplace}(0,\frac{1}{\varepsilon m})$.

The \emph{overall} Algorithm \ref{DworkDE} satisfies $2\varepsilon$-pan-privacy. Specifically, for all possible states (bitarrays) $\underline{\mathbf{b}}$ and outputs (estimated densities) $\underline{\tilde{d}}$, it holds that
\begin{equation*}
\begin{split}
\Prob\left(\mathbf{b(S)}=\underline{\mathbf{b}},\ \tilde{d}(S)=\underline{\tilde{d}} \right)
=&
\Prob\left(\mathbf{b(S)}=\underline{\mathbf{b}}\right)\\
& \Prob\left(\tilde{d}(S)=\underline{\tilde{d}} \ | \ \mathbf{b(S)}=\underline{\mathbf{b}} \right) \\
\leq & 
e^{\varepsilon} \Prob\left(\mathbf{b(S')}=\underline{\mathbf{b}}\right)
e^{\varepsilon} \\
& \Prob\left(\tilde{d}(S')=\underline{\tilde{d}} \ | \ \mathbf{b(S')}=\underline{\mathbf{b}} \right)\\
=&
e^{2\varepsilon}
\Prob \left( \mathbf{b(S')}=\underline{\mathbf{b}},\ \tilde{d}(S')=\underline{\tilde{d}} \right),
\end{split}
\end{equation*}

so the definition of pan-privacy is satisfied.
\end{proof}
\end{toappendix}

We next present two theorems on the accuracy of Algorithm \ref{DworkDE}.
The first theorem quantifies the bias and mean squared error of the estimator.
We note that Dwork et al. \cite{dwork2010pan} demonstrate that the estimator is unbiased, but do not discuss its mean squared error.

Let $S_{\mathcal{M}}$ be the subsequence (sub-stream) of the original stream $S$ that consists only of updates that refer to users in $\mathcal{M}$. In particular, $S_{\mathcal{M}}$ is constructed as $S_{\mathcal{M}} = \cup_{i \in \mathcal{M}} \{ s_i \} \subseteq S$.
Then, the following theorem quantifies the bias and mean squared error of the estimator $\tilde{d}$ of the density of $S_{\mathcal{M}}$. 
\begin{thm}\label{thm_DworkDE_acc1}
For a fixed sample $\mathcal{M}$, Algorithm \ref{DworkDE} provides an unbiased estimate $\tilde{d}$ of the density of $S_{\mathcal{M}}$ and has mean squared error
$$\E [ \ (\tilde{d}-d(S_{\mathcal{M}}))^2 \ ] \ \leq \ \frac{2(2m+1)}{m^2\varepsilon^2}.$$
\end{thm}

\begin{toappendix}
\begin{proof}[Proof of Theorem~\ref{thm_DworkDE_acc1}]
We begin with the \emph{bias} computation. We examine the distribution of an arbitrary entry in $\textbf{b}$:
\[ 
b_i \sim \left\{
\begin{array}{ll}
      \text{Bernoulli}(\frac{1}{2}) , & \mathcal{M}[i] \not\in S_{\mathcal{M}} ,\\
      \text{Bernoulli}(\frac{1}{2}+\frac{\varepsilon}{4}) , & \mathcal{M}[i] \in S_{\mathcal{M}} . \\
\end{array} 
\right. 
\]
Note that the distribution of $b_i$ does not depend on the number of appearances of $\mathcal{M}[i]$;
it only depends on whether it appeared or not.

Now, let $\hat{d}=\frac{1}{m}\sum_{i=1}^{m}{b_i}$. Then,
\begin{equation*}
\begin{split}
\E[\hat{d}] 
=&
\frac{1}{m}\sum_{i=1}^{m}{\E[b_i]} \\
=&
\frac{1}{m}\sum_{i=1}^{m} \E[b_i|\mathcal{M}[i] \in S_{\mathcal{M}}] \Prob(\mathcal{M}[i] \in S_{\mathcal{M}}) \\
&+ \E[b_i|\mathcal{M}[i] \not\in S_{\mathcal{M}}] \Prob(\mathcal{M}[i] \not\in S_{\mathcal{M}}) \\
=&
\frac{1}{m}\sum_{i=1}^{m}{ (\frac{1}{2}+\frac{\varepsilon}{4})d(S_{\mathcal{M}}) \ + \ \frac{1}{2}(1-d(S_{\mathcal{M}})) } \\
=&
\frac{1}{2} + \frac{\varepsilon}{4}d(S_{\mathcal{M}})
\end{split}
\end{equation*}
where we interpret the probability that a user is present in the sub-stream as the density of sub-stream.
This is true, if all users are considered equally likely to appear.

The final estimate (output) $\tilde{d}$ is then computed as  $\tilde{d} = \frac{4}{\varepsilon}(\hat{d}-\frac{1}{2}) + \text{Laplace}(0,\frac{1}{\varepsilon m})$, which gives that
$$ \E[\tilde{d}] \ = \ \frac{4}{\varepsilon}(\E[\hat{d}]-\frac{1}{2}) + \E[\text{Laplace}(0,\frac{1}{\varepsilon m})] \ = \ d(S_{\mathcal{M}}) ,$$
so $\tilde{d}$ is indeed an unbiased estimate.\\

We proceed with the \emph{mean squared error}. Given that $\tilde{d}$ is an unbiased estimate of $d(S_{\mathcal{M}})$, its mean squared error coincides with its variance.

A special case of the law of total variance states that, for any event $A$ and any random variable $b$, 
\begin{equation*}
\begin{split}
\text{var}(b) =& 
\text{var}(b|A) \Prob(A) 
+ \text{var}(b|A^C) \Prob(A^C) \\
&+ ( \E[b|A] - \E[b|A^C] )^2 \Prob(A) \Prob(A^C).
\end{split}
\end{equation*}
This property is used in conjunction with the fact that, if $b \sim \text{Bernoulli}(p)$, then $\text{var}(b)=p(1-p)$.
In our setting, we have $m$ Bernoulli random variables $b_i \ (i=1,...,m)$ and, for each user $i$, $A$ is the event that $\mathcal{M}[i] \in S_{\mathcal{M}}$ and $B$ is the event that $\mathcal{M}[i] \not\in S_{\mathcal{M}}$. Therefore, $b_i|A \sim \text{Bernoulli}(\frac{1}{2}+\frac{\varepsilon}{4})$ and $b_i|B \sim \text{Bernoulli}(\frac{1}{2})$. Then,
\begin{equation*}
\begin{split}
\text{var}(b_i) 
& = \ \frac{1}{4} - \frac{\varepsilon^2 d(S_{\mathcal{M}})^2}{16} \\
\Rightarrow
\text{var}(\hat{d}) 
& = \ \frac{1}{m^2}\sum_{i=1}^{m}{ \text{var}(b_i)} = \frac{1}{4m} - \frac{\varepsilon^2 d(S_{\mathcal{M}})^2}{16m} \\
\Rightarrow 
\text{var}(\tilde{d}) 
& = \ (\frac{4}{\varepsilon})^2 \text{var}(\hat{d}) + \text{var}(\text{Laplace}(0,\frac{1}{\varepsilon m}))\\
&= \frac{4}{m \varepsilon^2} - \frac{d(S_{\mathcal{M}})^2}{m} + \frac{2}{m^2 \varepsilon^2}\\
& \leq \frac{2(2m+1)}{m^2\varepsilon^2}
\end{split}
\end{equation*}
which completes the proof. To derive the last inequality we used the fact that $0\leq d(S_{\mathcal{M}}) \leq 1$.
\end{proof}
\end{toappendix}

The next theorem validates that the estimator provides the desired $(\alpha,\beta)$-approximation of the actual stream density. In contrast to Dwork et al. \cite{dwork2010pan}, we parameterize the proof, so we are then able to numerically compute the tightest version of the bound we derive.
\begin{thm}\label{thm_DworkDE_acc2}
If the sample maintained by Algorithm \ref{DworkDE} consists of $m = \mathcal{O}(\frac{1}{\varepsilon^2 \alpha^2}\log{\frac{1}{\beta}})$ users from $\mathcal{U}$, then, for fixed input $S$,
$\Prob\left(| \tilde{d}-d(S) | \ \geq \ \alpha \right)\ \leq \ \beta$
where the probability space is over the random choices of the algorithm.
\end{thm}

\begin{proofsketch}
Let $\hat{d}=\frac{1}{m}\sum_{i=1}^{m}{b_i}$. For some $\alpha>0$ and $\delta_1, \delta_2 \in (0,1)$,
\begin{equation*}
\begin{split}
\Prob \left(|  \right. \left. \tilde{d} - d(S) | \ \geq \ \alpha \right)
\leq&
\Prob\left(|\frac{\varepsilon}{4}\text{Laplace}(0,\frac{1}{\varepsilon m}) | \ \geq \ \frac{\varepsilon}{4}\delta_1 \delta_2 \alpha \right) \\
&+
\Prob\left(| \hat{d}-\E [\hat{d}] | \ \geq \ \frac{\varepsilon}{4}\delta_1 (1-\delta_2) \alpha\right) \\ 
&+ 
\Prob\left(|d(S_{\mathcal{M}})-d(S) | \ \geq \ (1-\delta_1)\alpha \right) \\
\triangleq& \ p_3+p_2+p_1.
\end{split}
\end{equation*}
For fixed input $S$, $d(S)$ is deterministic, $d(S_{\mathcal{M}})$ is random (due to sampling), $\hat{d}$ is random (as a sum of Bernoulli random variables), $\E[\hat{d}]$ is random (as a function of $d(S_{\mathcal{M}})$) and $\tilde{d}$ is random (as the sum of $\hat{d}$ and a Laplace random variable). We want $p_1+p_2+p_3 \leq \beta$. We bound each error probability separately, so for some $\delta_3, \delta_4 \in (0,1), \text{ such that } \delta_3 + \delta_4 < 1 $, we want $p_1<\delta_3 \beta , \ p_2<\delta_4 \beta , \ p_3<(1 - \delta_3 - \delta_4) \beta$. We obtain
\begin{equation*}
p_1 \leq \ 2 e^{-2m\alpha^2(1-\delta_1)^2}, \ p_2 \leq  2 e^{- \frac{1}{8} m \varepsilon^2 \alpha^2 \delta_1^2 (1-\delta_2)^2}, \ p_3 = e^{-\varepsilon \alpha m \delta_1 \delta_2}
\end{equation*}
and, therefore,
\begin{equation*}
p_1 \leq \delta_3 \beta \ \Leftrightarrow \ m \geq \underbrace{\frac{1}{2 \alpha^2 (1-\delta_1)^2} \log(\frac{2}{\beta \delta_3})}_{m_1} \ = \ \mathcal{O}(\frac{\log(\frac{1}{\beta})}{\alpha^2}).
\end{equation*}
\begin{equation*}
p_2 \leq \delta_4 \beta \ \Leftrightarrow \ m \geq \underbrace{\frac{8 \log(\frac{2}{\beta \delta_4})}{\varepsilon^2 \alpha^2 \delta_1^2 (1-\delta_2)^2}}_{m_2}
\ = \ \mathcal{O}(\frac{\log(\frac{1}{\beta})}{\varepsilon^2 \alpha^2}).
\end{equation*}
\begin{equation*}
\begin{split}
p_3 \leq (1-\delta_3-\delta_4) \beta \ \Leftrightarrow \ m & \geq \underbrace{\frac{1}{\varepsilon \alpha \delta_1 \delta_2} \log(\frac{1}{\beta (1-\delta_3-\delta_4)})}_{m_3} \\
& = \mathcal{O}(\frac{\log(\frac{1}{\beta})}{\varepsilon \alpha}).
\end{split}
\end{equation*}
By picking $m \geq \max\{m_1,m_2,m_3\}=\mathcal{O}(\frac{1}{\varepsilon^2 \alpha^2}\log(\frac{1}{\beta}))$, we ensure that
\begin{equation*}
\begin{split}
\Prob\left(| \tilde{d}-d(S) | \ \geq \ \alpha \right) & \leq p_1+p_2+p_3 \\
& \leq \delta_3 \beta+\delta_4 \beta+(1-\delta_3-\delta_4) \beta \ = \ \beta,
\end{split}
\end{equation*}
which completes the proof.
\end{proofsketch}

\begin{toappendix}
\begin{proof}[Proof of Theorem~\ref{thm_DworkDE_acc2}]
Before presenting the proof of the theorem, we state the following useful lemma.
\begin{lemma}\label{lemma_absolute_sum}
For any random variables $X$ and $Y$ and for some $\alpha>0$ and $\delta \in (0,1)$,
\begin{equation*}
\Prob( |X+Y|>\alpha ) \ \leq \ \Prob ( |X|>\alpha \delta ) + \Prob( |Y| >\alpha (1-\delta) ).
\end{equation*}
\end{lemma}
\begin{proof}[Proof of Lemma~\ref{lemma_absolute_sum}]
We have that
\begin{equation*}
\begin{split}
\{ (X,Y): |X+Y|>\alpha \} 
&= 
\{ (X,Y): X+Y>\alpha \text{ or } X+Y<-\alpha \} \\
&= 
\{ (X,Y): X+Y>\alpha \} \cup \{ (X,Y): X+Y<-\alpha \} \\
&\subseteq 
\{ (X,Y): X>\alpha \delta \text{ or }  Y >\alpha (1-\delta) \} \\
& 
\ \ \ \ \ \cup \{ (X,Y): X<-\alpha \delta \text{ or }  Y <-\alpha (1-\delta) \} \\
&= 
\{ (X,Y): X>\alpha \delta \text{ or }  Y >\alpha (1-\delta) \\
&
\ \ \ \ \ \text{ or } X<-\alpha \delta \text{ or }  Y <-\alpha (1-\delta) \} \\
&=
\{ (X,Y): |X|>\alpha \delta \text{ or } |Y| >\alpha (1-\delta) \}
\end{split}
\end{equation*}
so it follows that
\begin{equation*}
\begin{split}
\Prob( |X+Y|>\alpha ) 
& \leq 
\Prob ( |X|>\alpha \delta \text{ or } |Y| >\alpha (1-\delta) ) \\
& \leq 
\Prob ( |X|>\alpha \delta ) + \Prob( |Y| >\alpha (1-\delta) )
\end{split}
\end{equation*}
where the last inequality follows from the union bound.
\end{proof}

Let $\hat{d}=\frac{1}{m}\sum_{i=1}^{m}{b_i}$. We apply Lemma \ref{lemma_absolute_sum} twice, so, for some $\alpha>0$ and $\delta_1, \delta_2 \in (0,1)$,
\begin{equation*}
\begin{split}
\Prob \left(|  \right. & \left. \tilde{d} - d(S) | \ \geq \ \alpha \right)
\leq
\Prob\left(| \tilde{d}-d(S_{\mathcal{M}}) | \ \geq \ \delta_1 \alpha \right) \\ 
&+ 
\Prob\left(|d(S_{\mathcal{M}})-d(S) | \ \geq \ (1-\delta_1)\alpha \right) \\
=&
\Prob\left(| \frac{4}{\varepsilon}\hat{d}-\frac{2}{\varepsilon}+\text{Laplace}(0,\frac{1}{\varepsilon m}) - d(S_{\mathcal{M}}) | \ \geq \ \delta_1 \alpha \right) \\ 
&+ 
\Prob\left(|d(S_{\mathcal{M}})-d(S) | \ \geq \ (1-\delta_1)\alpha \right) \\
=&
\Prob\left(| \hat{d}-\frac{1}{2}+\frac{\varepsilon}{4}\text{Laplace}(0,\frac{1}{\varepsilon m})-\frac{\varepsilon}{4} d(S_{\mathcal{M}}) | \ \geq \ \frac{\varepsilon}{4}\delta_1 \alpha \right) \\ 
&+ 
\Prob\left(|d(S_{\mathcal{M}})-d(S) | \ \geq \ (1-\delta_1)\alpha \right) \\
=&
\Prob\left(| \hat{d}-\E [\hat{d}] +\frac{\varepsilon}{4}\text{Laplace}(0,\frac{1}{\varepsilon m}) | \ \geq \ \frac{\varepsilon}{4}\delta_1 \alpha \right) \\ 
&+ 
\Prob\left(|d(S_{\mathcal{M}})-d(S) | \ \geq \ (1-\delta_1)\alpha \right) \\
\leq&
\Prob\left(|\frac{\varepsilon}{4}\text{Laplace}(0,\frac{1}{\varepsilon m}) | \ \geq \ \frac{\varepsilon}{4}\delta_1 \delta_2 \alpha \right) \\
&+
\Prob\left(| \hat{d}-\E [\hat{d}] | \ \geq \ \frac{\varepsilon}{4}\delta_1 (1-\delta_2) \alpha\right) \\ 
&+ 
\Prob\left(|d(S_{\mathcal{M}})-d(S) | \ \geq \ (1-\delta_1)\alpha \right).
\end{split}
\end{equation*}

We examine the quantities involved in the final expression.
For fixed input $S$, $d(S)$ is deterministic.
On the contrary, $d(S_{\mathcal{M}})$ is random (due to sampling), $\hat{d}$ is random (as a sum of Bernoulli random variables), $\E[\hat{d}]$ is random (as a function of $d(S_{\mathcal{M}})$) and $\tilde{d}$ is random (as the sum of $\hat{d}$ and a Laplace random variable).
We therefore have three sources of error to control, which correspond to the following probabilities.
\begin{equation*}
\begin{aligned}
p_1 
& = 
\Prob\left(|d(S_{\mathcal{M}})-d(S) | \ \geq \ (1-\delta_1)\alpha \right), \\
p_2 
& = 
\Prob\left(| \hat{d}-\E [\hat{d}] | \ \geq \ \frac{\varepsilon}{4}\delta_1 (1-\delta_2) \alpha\right), \\
p_3 
& = 
\Prob\left(|\frac{\varepsilon}{4}\text{Laplace}(0,\frac{1}{\varepsilon m}) | \ \geq \ \frac{\varepsilon}{4}\delta_1 \delta_2 \alpha \right).
\end{aligned}
\end{equation*}

We want $p_1+p_2+p_3 \leq \beta$. We bound each error probability separately, so for some $\delta_3, \delta_4 \in (0,1), \text{ such that } \delta_3 + \delta_4 < 1 $, we want $p_1<\delta_3 \beta , \ p_2<\delta_4 \beta , \ p_3<(1 - \delta_3 - \delta_4) \beta$. \\

$\bullet$ \underline{First, we bound $p_1$.}\\
We define, $\forall i \in \{1,...,m\}$,
\[ 
X_i = \left\{
\begin{array}{ll}
      0 , & \mathcal{M}[i] \not\in S_{\mathcal{M}},\\
      1 , & \mathcal{M}[i] \in S_{\mathcal{M}}. \\
\end{array} 
\right. 
\]
and, since $S_{\mathcal{M}}$ was sampled uniformly at random from $S$, $X_i \sim \text{Bernoulli}(d(S))$. Since $d(S_{\mathcal{M}}) = \frac{1}{m}\sum_{i=1}^m{X_i}$, we obtain $\E[d(S_{\mathcal{M}})]=d(S)$.

Also, since $d(S_{\mathcal{M}})$ is a weighted sum of i.i.d. Bernoulli random variables, we take an additive, two-sided Chernoff bound that gives
\begin{equation*}
\begin{split}
p_1 &= \Prob\left(|\frac{1}{m}\sum_{i=1}^m{X_i}-\frac{1}{m}\sum_{i=1}^m{\E[X_i]} |  \geq (1-\delta_1)\alpha \right) \\
&\leq \ 2 e^{-2m\alpha^2(1-\delta_1)^2}
\end{split}
\end{equation*}
and, therefore,
\begin{equation*}
p_1 \leq \delta_3 \beta \ \Leftrightarrow \ m \geq \underbrace{\frac{1}{2 \alpha^2 (1-\delta_1)^2} \log(\frac{2}{\beta \delta_3})}_{m_1} \ = \ \mathcal{O}(\frac{\log(\frac{1}{\beta})}{\alpha^2}).
\end{equation*}

$\bullet$ \underline{Next, we bound $p_2$.}\\
As we already mentioned, both $\hat{d}$ and $\E[\hat{d}]=\frac{1}{2} + \frac{\varepsilon}{4}d(S_{\mathcal{M}})$ are random. $d(S_{\mathcal{M}})$ takes values in the set $\mathcal{D} = \{0,\frac{1}{m},...,\frac{m-1}{m},1\}$, so, using the law of total probability,
\begin{equation*}
\begin{split}
p_2 = & \sum_{\underline{d} \in \mathcal{D}} \Prob\left(| \hat{d}-\E[\hat{d}] | \ \geq \ \frac{\varepsilon}{4}\delta_1 (1-\delta_2) \alpha\ \ | \ d(S_{\mathcal{M}}) = \underline{d} \right) \\ 
& \ \ \Prob\left(d(S_{\mathcal{M}}) = \underline{d}\right) .
\end{split}
\end{equation*}

We make the following two critical remarks:\\
i) For fixed $d(S_{\mathcal{M}}) = \underline{d}$, $\E[\hat{d}]$ is no longer random.\\
ii) For $i\in \{1,...,m\}$, let $d_i$ denote the distribution from which each $b_i$ is drawn from, so that $d_i=0 \Rightarrow b_i \sim \text{Bernoulli}(\frac{1}{2})$ and $d_i = 1 \Rightarrow b_i \sim \text{Bernoulli}(\frac{1}{2}+\frac{\varepsilon}{4})$. Once we fix $d(S_{\mathcal{M}}) = \underline{d}$, the $d_i$'s are not independent, as $\sum_{i=1}^m{d_i}=m\underline{d}$.
However, the $b_i$'s are independent, as we impose no constraint on them and each is drawn independently from a fixed distribution. Therefore, $\hat{d}$ is a sum of independent Poisson trials.

We again make use of an additive, two-sided Chernoff bound, so
\begin{equation*}
\begin{split}
\Prob \left(| \hat{d}-\E[\hat{d}] | \right. & \left. \geq \ \frac{\varepsilon}{4}\delta_1 (1-\delta_2) \alpha\ \ | \ d(S_{\mathcal{M}}) = \underline{d} \right) \\
& \leq 2 e^{-2m \varepsilon^2 \alpha^2 \delta_1^2 (1-\delta_2)^2 \frac{1}{16}}.
\end{split}
\end{equation*}

We observe that the bound is independent of $\underline{d}$, so
\begin{equation*}
\begin{split}
p_2 & \leq \sum_{\underline{d} \in \mathcal{D}}{ 2 e^{-2m \varepsilon^2 \alpha^2 \delta_1^2 (1-\delta_2)^2 \frac{1}{16}} \ \Prob\left(d(S_{\mathcal{M}}) = \underline{d}\right) } \\
& = 2 e^{- \frac{1}{8} m \varepsilon^2 \alpha^2 \delta_1^2 (1-\delta_2)^2}
\end{split}
\end{equation*}
and, therefore,
\begin{equation*}
p_2 \leq \delta_4 \beta \ \Leftrightarrow \ m \geq \underbrace{\frac{8 \log(\frac{2}{\beta \delta_4})}{\varepsilon^2 \alpha^2 \delta_1^2 (1-\delta_2)^2}}_{m_2}
\ = \ \mathcal{O}(\frac{\log(\frac{1}{\beta})}{\varepsilon^2 \alpha^2}).
\end{equation*}

$\bullet$ \underline{Finally, we compute $p_3$.}
\begin{equation*}
\begin{split}
p_3 
&= 
\Prob\left(|\frac{\varepsilon}{4}\text{Laplace}(0,\frac{1}{\varepsilon m}) | \ \geq \ \frac{\varepsilon}{4}\delta_1 \delta_2 \alpha \right) \\
&= 
\Prob\left(|\text{Laplace}(0,\frac{1}{\varepsilon m}) | \ \geq \ \delta_1 \delta_2 \alpha \right) \\
&=
\Prob\left(\text{Exponential}(\varepsilon m) \ \geq \ \delta_1 \delta_2 \alpha \right) \\
&=
e^{-\varepsilon \alpha m \delta_1 \delta_2}
\end{split}
\end{equation*}
and, therefore,
\begin{equation*}
\begin{split}
p_3 \leq (1-\delta_3-\delta_4) \beta \ \Leftrightarrow \ m & \geq \underbrace{\frac{1}{\varepsilon \alpha \delta_1 \delta_2} \log(\frac{1}{\beta (1-\delta_3-\delta_4)})}_{m_3} \\
& = \mathcal{O}(\frac{\log(\frac{1}{\beta})}{\varepsilon \alpha}).
\end{split}
\end{equation*}

By picking $m \geq \max\{m_1,m_2,m_3\}=\mathcal{O}(\frac{1}{\varepsilon^2 \alpha^2}\log(\frac{1}{\beta}))$, we ensure that
\begin{equation*}
\begin{split}
\Prob\left(| \tilde{d}-d(S) | \ \geq \ \alpha \right) & \leq p_1+p_2+p_3 \\
& \leq \delta_3 \beta+\delta_4 \beta+(1-\delta_3-\delta_4) \beta \ = \ \beta,
\end{split}
\end{equation*}
which completes the proof.
\end{proof}
\end{toappendix}

Theorem \ref{thm_DworkDE_acc2} offers an additive error guarantee, which may not be so useful if the density of the input stream is small.
Dwork et al. \cite{dwork2010pan} show how to modify their algorithm to obtain a multiplicative error guarantee, which is more meaningful in such cases.
We do not examine this point in our work.
Additionally, as we already stated, our parameterized proof allows us to optimally tune the parameters $\delta_1,\delta_2,\delta_3,\delta_4$ and, as a result, compute the tightest version of the bound we derive.
In particular, for fixed $\varepsilon,$ $\alpha,$ and $m$, the tightest bound $\beta$ on $\Prob\left(| \tilde{d}-d(S) | \ \geq \ \alpha \right)$ is computed by numerically solving the following optimization problem.

\begin{equation*}
\begin{aligned}
&\underset{\delta_1,\delta_2}{\text{minimize}} & 
& \beta(\delta_1,\delta_2)
= 
2 e^{-2m\alpha^2(1-\delta_1)^2}
+ 2 e^{- \frac{1}{8} m \varepsilon^2 \alpha^2 \delta_1^2 (1-\delta_2)^2}
\\
&&& + e^{-\varepsilon \alpha m \delta_1 \delta_2}
\\
& \text{subject to} &
& 0 \leq \delta_1 \leq 1 , \ 
0 \leq \delta_2 \leq1 .
\end{aligned}
\end{equation*}

\subsection{Picking the Optimal Sample Size}\label{subSEC_sample_size}
Theorem \ref{thm_DworkDE_acc2} provides an asymptotic expression for the sample size $m$ to achieve the desired approximation accuracy.
A question that arises is how we should pick $m$ in practice.
This is again achieved by taking advantage of our parameterized proof and numerically solving a similar optimization problem,
which allows us to compute the optimal (minimum) sample size $m^*$ that achieves the desired approximation accuracy, according to the bounds we derived.
Specifically, for fixed $\varepsilon,$ $\alpha,$ and $\beta$, we define
\begin{equation*}
\begin{split}
m(\delta_1,\delta_2,\delta_3,\delta_4) 
=&
\max \{ \ m_1 , m_2 , m_3 \ \} \\
=&
\max \{ 
\ \frac{\log(\frac{2}{\beta \delta_3})}{2 \alpha^2 (1-\delta_1)^2} ,
\ \frac{8 \log(\frac{2}{\beta \delta_4})}{\varepsilon^2 \alpha^2 \delta_1^2 (1-\delta_2)^2} ,\\
& \ \ \frac{\log(\frac{1}{\beta (1-\delta_3-\delta_4)})}{\varepsilon \alpha \delta_1 \delta_2} \ \}.
\end{split}
\end{equation*}
Then, we pick $m^*$ as the solution to the following optimization problem.
\begin{equation*}
\begin{aligned}
& \underset{\delta_1,\delta_2,\delta_3,\delta_4}{\text{minimize}} 
& &m(\delta_1,\delta_2,\delta_3,\delta_4) \\
& \text{subject to}
& & 0\leq\delta_i\leq1 , \ i \in \{1,2,3,4\},\\
& & & \delta_3+\delta_4\leq1 .
\end{aligned}
\end{equation*}

Figure \ref{FIG_sample_size_bound} illustrates the proposed sample size $m$ as a function of the privacy budget $\varepsilon$.
We recall that our approach allows us to compute the tightest version of the specific bound that we derive on the probability of error and, hence, on the sample size.
Nevertheless, we remark that the bound itself is not tight;
this is an experimental observation and indicates that the same accuracy can be achieved with even fewer samples.

\begin{figure}
    \centering
    \includegraphics[width=\columnwidth]{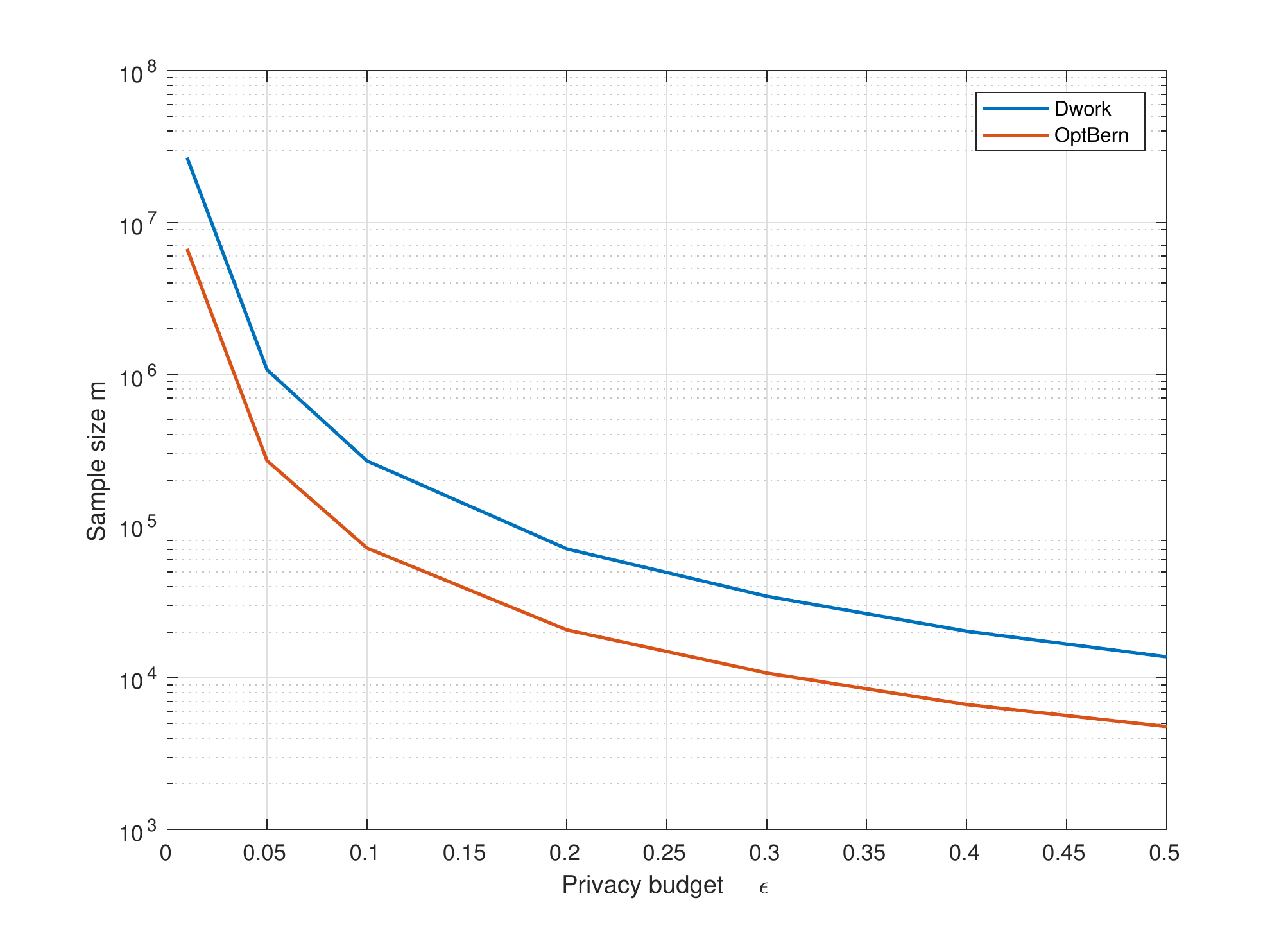}
    \caption{Proposed sample size for \texttt{Dwork} and \texttt{OptBern}.}
	\label{FIG_sample_size_bound}
\end{figure}

\section{Improved Pan-Private Density Estimator}
In this section, we modify Algorithm \ref{DworkDE} and derive a novel algorithm that significantly outperforms the original one (both theoretically and experimentally).
The key reason behind our algorithm's superiority is that, in contrast to Algorithm \ref{DworkDE}, it manages to use all the allocated privacy budget.

\subsection{On the Use of the Allocated Privacy Budget}
Recall that, to ensure that its state satisfies differential privacy, Algorithm \ref{DworkDE} utilizes two different distributions, one for users that do not appear in the stream and one for users that do appear.
We now introduce a little extra notation; the bit that corresponds to a user from the former category is drawn from the distribution with pmf $f_{\text{init}} = \text{Bernoulli}(\frac{1}{2})$, while the bit that corresponds to a user from the latter category is drawn from $f_{\text{upd}} = \text{Bernoulli}(\frac{1}{2}+\frac{\varepsilon}{4})$.
Although by $f_{\text{init}}$ and $f_{\text{upd}}$ we formally denote the probability mass functions of the two Bernoulli distributions, at some points we use the same notation to refer to the distributions themselves.
To satisfy differential privacy, we have to ensure that, $\forall b \in \{0,1\}$,
$$e^{-\varepsilon} \leq R(b)=\frac{f_{\text{upd}}(b)}{f_{\text{init}}(b)} \leq e^{\varepsilon}.$$
Figure \ref{FIG_privacy_budget} presents the ratios $R(0)$ and $R(1)$ for Algorithm \ref{DworkDE} as a function of the privacy budget $\varepsilon$.
Although Algorithm \ref{DworkDE} does ensure that its state satisfies differential privacy (as we have already proved in Theorem \ref{thm_DworkDE_priv}), it fails to use all the allocated privacy budget.
\begin{figure}
    \centering
    \includegraphics[width=0.75\columnwidth]{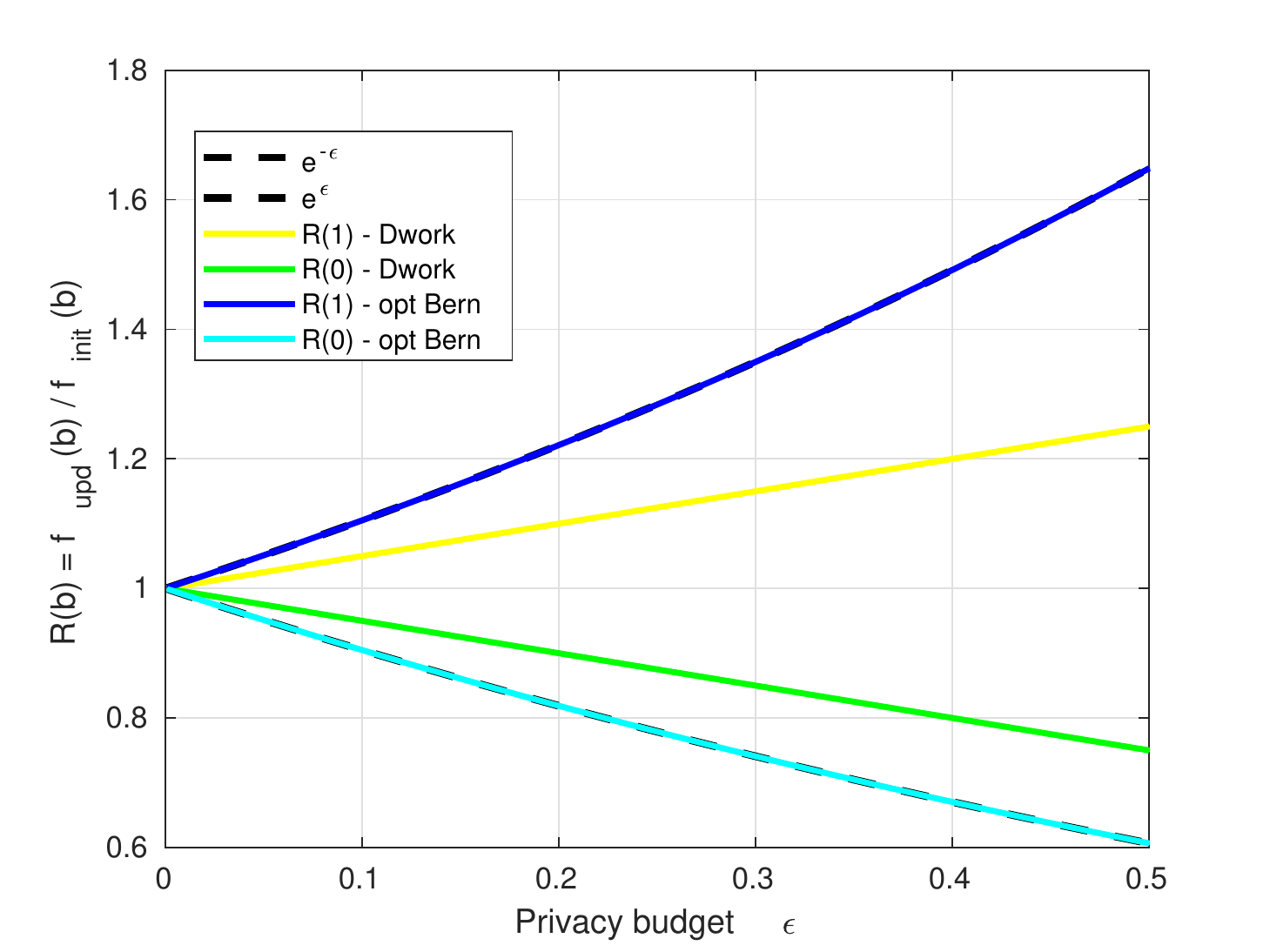}
    \caption{State differential privacy of \texttt{Dwork} and \texttt{OptBern}.}
	\label{FIG_privacy_budget}
\end{figure} 
Indeed, we show (in the proof of Theorem \ref{thm_DworkDE_priv} in the Appendix) that $R(0)=1-\frac{\varepsilon}{2}$ and $R(1)=1+\frac{\varepsilon}{2}$. Let $\varepsilon_{\text{actual}}>0$ be the actual privacy budget that Algorithm \ref{DworkDE} consumes. Then,
\[ 
\left.
\begin{array}{ll}
     e^{-\varepsilon_{\text{actual}}} \leq 1-\frac{\varepsilon}{2} \leq e^{\varepsilon_{\text{actual}}}\\
     e^{-\varepsilon_{\text{actual}}} \leq 1+\frac{\varepsilon}{2} \leq e^{\varepsilon_{\text{actual}}} \\
\end{array} 
\right\}
\Rightarrow
\]
$${\varepsilon_{\qquad \text{actual}}} \geq \max \left\{\ \log(1+\frac{\varepsilon}{2}) , \ -\log(1-\frac{\varepsilon}{2}) \ \right\}.$$
In Figure \ref{FIG_actual_privacy_budget}, the actual privacy budget $\varepsilon_{\text{actual}}$ used by Algorithm \ref{DworkDE} is plotted as a function of the allocated privacy budget $\varepsilon_{\text{allocated}} = \varepsilon$.

\begin{figure}
    \centering
    \includegraphics[width=0.75\columnwidth]{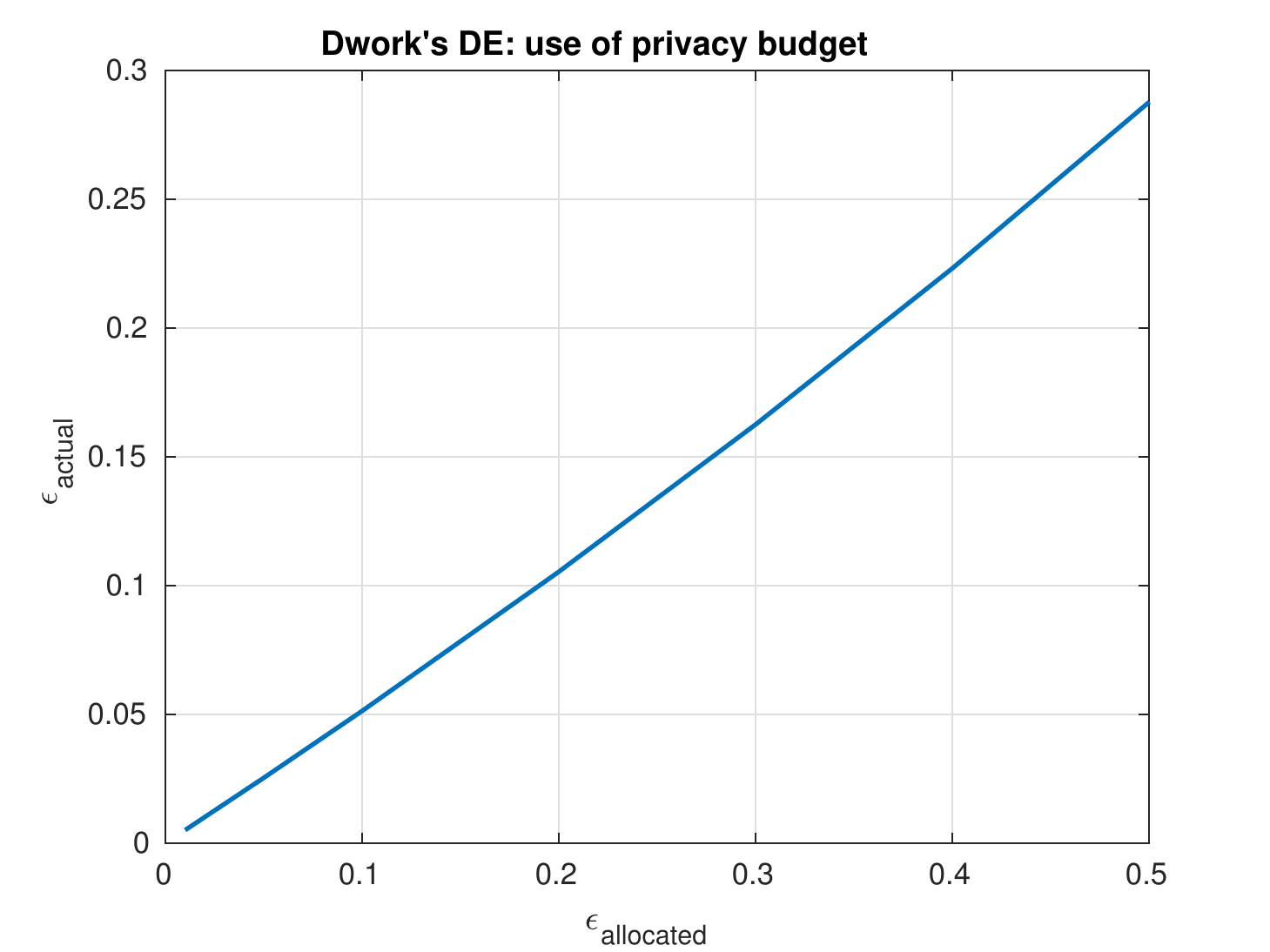}
    \caption{Allocated vs actual privacy budget of \texttt{Dwork}.}
	\label{FIG_actual_privacy_budget}
\end{figure}

\subsection{Optimally Tuning the Bernoulli Distributions} \label{sec::bernoulli-optimal-tuning}
Based on the aforementioned observation, for any given $\varepsilon$, we optimally tune the Bernoulli distributions used (with pmf's $f_{\text{init}}$ and $f_{\text{upd}}$), by picking a pair of parameters that tightly satisfies $\varepsilon$-differential privacy.
In particular, we maximize the distributions' distance (or, more precisely, the difference of their parameters), which allows us to more accurately distinguish between users that did not appear and users that appeared in the stream during the density estimation computation (without sacrificing the users' privacy!).

We propose that the parameters of $f_{\text{init}}$ and $f_{\text{upd}}$ be picked symmetric around some value $c$; that is, we end up with the distributions $f_{\text{init}} = \text{Bernoulli}(c-x)$, $f_{\text{upd}} = \text{Bernoulli}(c+x),$
for some $x > 0$ that corresponds to half the difference of the distributions' parameters. 
Apparently, we want $0<c-x<c+x<1$.
For example, in Algorithm \ref{DworkDE}, $c=\frac{1}{2}+\frac{\varepsilon}{8}$ and $x=\frac{\varepsilon}{8}$, so that $0 < c-x=\frac{1}{2} < c+x=\frac{1}{2}+\frac{\varepsilon}{4} < 1 , \ \forall \varepsilon \in (0,\frac{1}{2}]$. 

The modification we propose is in the selection of $x$.
Again, to satisfy differential privacy, we have to ensure that, $\forall b \in \{0,1\}$,
\begin{equation*}
\begin{split}
e^{-\varepsilon} \leq \frac{f_{\text{upd}}(b)}{f_{\text{init}}(b)} \leq e^{\varepsilon}
& \Leftrightarrow
\left\{
\begin{array}{ll}
     e^{-\varepsilon} \leq \frac{c+x}{c-x} \leq e^{\varepsilon} \\
     e^{-\varepsilon} \leq \frac{c-x}{c+x} \leq e^{\varepsilon} \\
\end{array} 
\right.\\
& \Leftrightarrow
x \ \leq \ \frac{e^{\varepsilon}-1}{e^{\varepsilon}+1} \ c
\ = \ \tanh(\frac{\varepsilon}{2}) c .
\end{split}
\end{equation*}
Since our goal was to maximize the difference of the distributions' parameters, we pick $x=\tanh(\frac{\varepsilon}{2})c$.
The proposed distributions are
\begin{equation*}
\begin{split}
f_{\text{init}} = \text{Bernoulli}\left(c \ (1-\tanh(\frac{\varepsilon}{2})\right), \\
f_{\text{upd}} = \text{Bernoulli}\left(c \ (1+\tanh(\frac{\varepsilon}{2})\right),
\end{split}
\end{equation*}
and, since both parameters have to be greater than zero and less than one, we can determine the values that $c$ can take for each fixed value of $\varepsilon$. For $0 < \varepsilon \leq \frac{1}{2}$, it is easy to see that $0 < \tanh(\frac{\varepsilon}{2}) < \frac{1}{4}$ (since $\tanh(\frac{\varepsilon}{2})$ is a monotonically increasing function of $\varepsilon$), so we can pick any $c \in (0,\frac{4}{5}]$, regardless of the specific value of $\varepsilon \in (0,\frac{1}{2}]$.
Indeed, Figure \ref{FIG_privacy_budget} illustrates that the proposed Bernoulli distributions use all the allocated privacy budget.

\subsection{Estimator \& Analysis}
We now present the modified algorithm (Algorithm \ref{OptBernDE}), which we call \verb|OptBern| (Optimal Bernoulli Density Estimator). For simplicity, we set $c=\frac{1}{2}$.

\begin{algorithm}

\caption{\texttt{OptBern} \label{OptBernDE}}

\DontPrintSemicolon

\KwIn{Data stream $S$, Privacy budget $\varepsilon$, Accuracy parameters $(\alpha,\beta)$}

\KwOut{Density $\tilde{d}(S)$}

Compute $m^*$ and set $m=m^*$\;

Sample a random subset $\mathcal{M} \subseteq \mathcal{U}$ of $m$ users (without replacement) and define an arbitrary ordering over $\mathcal{M}$\;

Create a bitarray $\mathbf{b}=[b_1\ ...\ b_{m}]$ and map $\mathcal{M}[i] \rightarrow b_i,\ \forall \ i \in \{1,...,m\}$\;

Initialize $\mathbf{b}$ randomly: $b_i \sim \text{Bernoulli}(\frac{1}{2}(1-\tanh(\frac{\varepsilon}{2}))), \ \forall \ i \in \{1,...,m\}$\;

\For{t = 1 \KwTo T}{

	\If{$s_t \in \mathcal{M}$}{
		
		Find $i:\ \mathcal{M}[i]=s_t$\;
		
		Re-sample: $b_i \sim \text{Bernoulli}(\frac{1}{2}(1+\tanh(\frac{\varepsilon}{2})))$\;
			
	}
}

Return $\tilde{d}(S)\ = \ \frac{1}{\tanh(\frac{\varepsilon}{2})}(\frac{1}{m}\sum_{i=1}^{m}{b_i}-\frac{1}{2} + \frac{1}{2} \tanh(\frac{\varepsilon}{2}) ) \ + \ \text{Laplace}(0,\frac{1}{\varepsilon m})$\;

\end{algorithm}

Next, we proceed with the privacy and accuracy analysis of Algorithm \ref{OptBernDE}, following the lines of our analysis of Algorithm \ref{DworkDE}.

\begin{thm}\label{thm_OptBernDE_priv}
If $\varepsilon \leq \frac{1}{2}$, then Algorithm \ref{OptBernDE} satisfies $2\varepsilon$-pan-privacy and utilizes all the allocated privacy budget.
\end{thm}

\begin{toappendix}
\begin{proof}[Proof of Theorem~\ref{thm_OptBernDE_priv}]
The proof is identical to that of Theorem \ref{thm_DworkDE_priv}.
The only modification is on proving that the state (bitarray) satisfies $\varepsilon$-differential privacy. In particular, for a user $u \in \mathcal{M}$ that appears in stream $S$ and does not appear in stream $S'$ (again, let $b_1$ be the entry that corresponds to $u$ in the bitarray), we have
\begin{equation*}
\begin{split}
\frac{\Prob\left(\mathbf{b(S)}=[1 \ \underline{b_2} \ ...\ \underline{b_{m}}]\right)}{\Prob\left(\mathbf{b(S')}=[1 \ \underline{b_2} \ ...\ \underline{b_{m}}]\right)} 
&=
\frac{\Prob\left(b_1(S)=1\right)}{\Prob\left(b_1(S')=1\right)} \\
&=
\ \frac{ \frac{1}{2}(1+\frac{e^{\varepsilon}-1}{e^{\varepsilon}+1}) }{ \frac{1}{2}(1-\frac{e^{\varepsilon}-1}{e^{\varepsilon}+1}) } 
= 
\ e^{\varepsilon} , \\
\frac{\Prob\left(\mathbf{b(S)}=[0 \ \underline{b_2} \ ...\ \underline{b_{m}}]\right)}{\Prob\left(\mathbf{b(S')}=[0 \ \underline{b_2} \ ...\ \underline{b_{m}}]\right)} 
&= 
\frac{\Prob\left(b_1(S)=0\right)}{\Prob\left(b_1(S')=0\right)} \\
&=
\ \frac{ \frac{1}{2}(1-\frac{e^{\varepsilon}-1}{e^{\varepsilon}+1}) }{ \frac{1}{2}(1+\frac{e^{\varepsilon}-1}{e^{\varepsilon}+1}) } 
= 
\ e^{-\varepsilon},
\end{split}
\end{equation*}
so user $u$ is guaranteed $\varepsilon$-differential privacy against an adversary that observes $u$'s entry in $\mathbf{b}$.
\end{proof}
\end{toappendix}

\begin{thm}\label{thm_OptBernDE_acc1}
For a fixed sample $\mathcal{M}$, Algorithm \ref{OptBernDE} provides an unbiased estimate $\tilde{d}$ of the density of $S_{\mathcal{M}}$ and has mean squared error
$$\E [ \ (\tilde{d}-d(S_{\mathcal{M}}))^2 \ ] \ \leq \ \frac{1}{4m \tanh^2(\frac{\varepsilon}{2})} + \frac{2}{m^2 \varepsilon^2}.$$
\end{thm}

\begin{toappendix}
\begin{proof}[Proof of Theorem~\ref{thm_OptBernDE_acc1}]
The proof is similar to that of Theorem \ref{thm_DworkDE_acc1}.\\

We begin with the \emph{bias} computation.
To examine the distribution of an arbitrary entry in $\textbf{b}$, we introduce the notation $p_{\text{init}}=\frac{1}{2}(1-\tanh(\frac{\varepsilon}{2}))$ and $p_{\text{upd}}=\frac{1}{2}(1+\tanh(\frac{\varepsilon}{2}))$. Then,
\[ 
b_i \sim \left\{
\begin{array}{ll}
      \text{Bernoulli}(p_{\text{init}}) , & \mathcal{M}[i] \not\in S_{\mathcal{M}},\\
      \text{Bernoulli}(p_{\text{upd}}) , & \mathcal{M}[i] \in S_{\mathcal{M}}. \\
\end{array} 
\right. 
\]

Denoting $\hat{d}=\frac{1}{m}\sum_{i=1}^{m}{b_i}$ and applying a computation similar to the one used in the proof of Theorem \ref{thm_DworkDE_acc1}, we obtain 
\begin{equation*}
\begin{split}
\E[\hat{d}] 
&=
p_{\text{init}} + (p_{\text{upd}}-p_{\text{init}})d(S_{\mathcal{M}})  \\
&= \frac{1}{2}(1-\tanh(\frac{\varepsilon}{2})) + \tanh(\frac{\varepsilon}{2}) d(S_{\mathcal{M}}).
\end{split}
\end{equation*}

The final estimate (output) $\tilde{d}$ is
\begin{equation*}
\begin{split}
\tilde{d}
&=
\frac{\hat{d}-p_{\text{init}}}{p_{\text{upd}}-p_{\text{init}}} + \text{Laplace}(0,\frac{1}{\varepsilon m}) \\
&=
\frac{1}{\tanh(\frac{\varepsilon}{2})}(\hat{d}-\frac{1}{2} + \frac{1}{2} \tanh(\frac{\varepsilon}{2}) ) \ + \ \text{Laplace}(0,\frac{1}{\varepsilon m}),
\end{split}
\end{equation*}
so $\tilde{d}$ is an unbiased estimate.\\

We proceed with the \emph{mean squared error}. Again, since $\tilde{d}$ is an unbiased estimate of $d(S_{\mathcal{M}})$, its mean squared error coincides with its variance.
Also, recall that we have $m$ Bernoulli random variables $b_i \ (i=1,...,m)$, so applying the law of total variance gives us
\begin{equation*}
\begin{split}
\text{var}(b_i) 
= & 
p_{\text{init}}(1-p_{\text{init}}) \\
& + 
d(S_{\mathcal{M}})[-p_{\text{init}}(1-p_{\text{init}}) + p_{\text{upd}}(1-p_{\text{upd}})] \\
& + 
d(S_{\mathcal{M}})(1-d(S_{\mathcal{M}}))(p_{upd}-p_{init})^2 \\
= &
\frac{1}{4}[1-\tanh^2(\frac{\varepsilon}{2})] \\
& + d(S_{\mathcal{M}})(1-d(S_{\mathcal{M}})) \tanh^2(\frac{\varepsilon}{2}) \\
\Rightarrow \text{var}(\hat{d}) 
= & 
\frac{1}{m^2}\sum_{i=1}^{m}{ \text{var}(b_i)} \\
= &
\frac{1}{4m}[1-\tanh^2(\frac{\varepsilon}{2})] \\
& + 
\frac{d(S_{\mathcal{M}})(1-d(S_{\mathcal{M}}))}{m} \tanh^2(\frac{\varepsilon}{2}) \\
\Rightarrow \text{var}(\tilde{d}) 
= &
(\frac{1}{p_{upd}-p_{init}})^2 \text{var}(\hat{d}) + \text{var}(\text{Laplace}(0,\frac{1}{\varepsilon m})) \\
= &
\frac{1}{4m}[\frac{1}{\tanh^2(\frac{\varepsilon}{2})}-1] \\
& + \frac{d(S_{\mathcal{M}})(1-d(S_{\mathcal{M}}))}{m} + \frac{2}{m^2 \varepsilon^2}\\
\leq &
\frac{1}{4m \tanh^2(\frac{\varepsilon}{2})} + \frac{2}{m^2 \varepsilon^2} ,
\end{split}
\end{equation*}
which completes the proof. To derive the last inequality, we used the fact that, since $0\leq d(S_{\mathcal{M}}) \leq 1$, it follows that $d(S_{\mathcal{M}})(1-d(S_{\mathcal{M}})) \leq \frac{1}{4}$.
\end{proof}
\end{toappendix}

\begin{thm}\label{thm_OptBernDE_acc2}
If the sample maintained by Algorithm \ref{OptBernDE} consists of $m = \mathcal{O}(\frac{1}{\varepsilon^2 \alpha^2}\log{\frac{1}{\beta}})$ users from $\mathcal{U}$, then, for fixed input $S$, 
$\Prob\left(| \tilde{d}-d(S) | \ \geq \ \alpha \right)\ \leq \ \beta$
where the probability space is over the random choices of the algorithm.
\end{thm}

\begin{toappendix}
\begin{proof}[Proof of Theorem~\ref{thm_OptBernDE_acc2}]
The proof is similar to that of Theorem \ref{thm_DworkDE_acc2}, so we only focus on their differences.\\

Let $\hat{d}=\frac{1}{m}\sum_{i=1}^{m}{b_i}$. We again apply Lemma \ref{lemma_absolute_sum} twice, so, for some $\alpha>0$ and $\delta_1, \delta_2 \in (0,1)$,
\begin{equation*}
\begin{split}
\Prob & \left(| \tilde{d}-d(S) | \ \geq \ \alpha \right)
\leq 
\Prob\left(| \tilde{d}-d(S_{\mathcal{M}}) | \ \geq \ \delta_1 \alpha \right) \\
& +
\Prob\left(|d(S_{\mathcal{M}})-d(S) | \ \geq \ (1-\delta_1)\alpha \right) \\
= &
\Prob\left( | 
\frac{\hat{d}}{\tanh(\frac{\varepsilon}{2})}
- \frac{1}{2\tanh(\frac{\varepsilon}{2})} 
+ \frac{1}{2} 
+ \text{Laplace}(0,\frac{1}{\varepsilon m}) \right. \\
& \left. - d(S_{\mathcal{M}}) 
| \geq \delta_1 \alpha \right) 
\ + \ \Prob\left(|d(S_{\mathcal{M}})-d(S) | \ \geq \ (1-\delta_1)\alpha \right) \\
= & 
\Prob\left(| \hat{d}-\E [\hat{d}] +\tanh(\frac{\varepsilon}{2})\text{Laplace}(0,\frac{1}{\varepsilon m}) | \geq \tanh(\frac{\varepsilon}{2})\delta_1 \alpha \right) \\ 
& + 
\ \Prob\left(|d(S_{\mathcal{M}})-d(S) | \ \geq \ (1-\delta_1)\alpha \right) \\
\leq &
\Prob\left(|\tanh(\frac{\varepsilon}{2}) \text{Laplace}(0,\frac{1}{\varepsilon m}) | \ \geq \ \tanh(\frac{\varepsilon}{2})\delta_1 \delta_2 \alpha \right) \\
& +
\ \Prob\left(| \hat{d}-\E [\hat{d}] | \ \geq \ \tanh(\frac{\varepsilon}{2}) \delta_1 (1-\delta_2) \alpha\right) \\ 
& + 
\ \Prob\left(|d(S_{\mathcal{M}})-d(S) | \ \geq \ (1-\delta_1)\alpha \right) \\
= &
p_3 + p_2 + p_1 \ \text{ (respectively)}.
\end{split}
\end{equation*}

As in the proof of Theorem \ref{thm_DworkDE_acc2}, we have three sources of error to control and we want $p_1+p_2+p_3 \leq \beta$. We bound each error probability separately, so, for some $\delta_3, \delta_4 \in (0,1), \text{ such that } \delta_3 + \delta_4 < 1 $, we want $p_1<\delta_3 \beta , \ p_2<\delta_4 \beta , \ p_3<(1 - \delta_3 - \delta_4) \beta$. \\

Bounding $p_1$ and $p_3$ is identical, as these error probabilities are unchanged. We copy the bounds we derived in Theorem \ref{thm_DworkDE_acc2}.
\begin{equation*}
\begin{split}
p_1 
& \leq 
2 e^{-2m\alpha^2(1-\delta_1)^2} \leq \delta_3 \beta \\
& \Leftrightarrow 
m \geq \underbrace{\frac{1}{2 \alpha^2 (1-\delta_1)^2} \log(\frac{2}{\beta \delta_3})}_{m_1} 
= 
\mathcal{O}(\frac{1}{\alpha^2}\log(\frac{1}{\beta})) , \\
p_3 
& \leq 
e^{-\varepsilon \alpha m \delta_1 \delta_2} \leq (1-\delta_3-\delta_4) \beta \\
& \Leftrightarrow 
m \geq \underbrace{\frac{1}{\varepsilon \alpha \delta_1 \delta_2} \log(\frac{1}{\beta (1-\delta_3-\delta_4)})}_{m_3} 
= 
\mathcal{O}(\frac{1}{\varepsilon \alpha}\log(\frac{1}{\beta})).
\end{split}
\end{equation*}

We focus on $p_2$.
The difference is due to the fact that we now have $\E[\hat{d}]=\frac{1}{2}(1-\tanh(\frac{\varepsilon}{2})) + \tanh(\frac{\varepsilon}{2})d(S_{\mathcal{M}})$. Although $\E[\hat{d}]$ has changed, it still is a random mean (as $d(S_{\mathcal{M}})$ is random), so we again apply the total probability theorem
\begin{equation*}
\begin{split}
p_2 = & \sum_{\underline{d} \in \mathcal{D}} \Prob\left(| \hat{d}-\E [\hat{d}] | \geq \tanh(\frac{\varepsilon}{2}) \delta_1 (1-\delta_2) \alpha | d(S_{\mathcal{M}}) = \underline{d} \right) \\
& \ \ \Prob\left(d(S_{\mathcal{M}}) = \underline{d}\right).
\end{split}
\end{equation*}

Based on observations identical to those we made when proving Theorem \ref{thm_DworkDE_acc2}, we again make use of an additive, two-sided Chernoff bound
\begin{equation*}
\begin{split}
\Prob\left( | \hat{d}-\E [\hat{d}] | \right.
& \left. \geq 
\tanh(\frac{\varepsilon}{2}) \delta_1 (1-\delta_2) \alpha 
| d(S_{\mathcal{M}}) = \underline{d} \right) \\
& \leq 2 e^{-2m \tanh^2(\frac{\varepsilon}{2}) \alpha^2 \delta_1^2 (1-\delta_2)^2}
\end{split}
\end{equation*}

and, since the bound is independent of $\underline{d}$,
\begin{equation*}
\begin{split}
p_2 & \leq \sum_{\underline{d} \in \mathcal{D}}{ 2 e^{-2m \tanh^2(\frac{\varepsilon}{2}) \alpha^2 \delta_1^2 (1-\delta_2)^2} \ \Prob\left(d(S_{\mathcal{M}}) = \underline{d}\right) } \\
&= 2 e^{-2m \tanh^2(\frac{\varepsilon}{2}) \alpha^2 \delta_1^2 (1-\delta_2)^2}
\end{split}
\end{equation*}

and, therefore, noting that $\tanh(\frac{x}{2})=\frac{x}{2}-\frac{x^3}{24}+\mathcal{O}(x^5)=\mathcal{O}(x)$,
\begin{equation*}
p_2 \leq \delta_4 \beta 
\ \Leftrightarrow 
\ m \geq \underbrace{\frac{ \log(\frac{2}{\beta \delta_4}) }{2 \tanh^2(\frac{\varepsilon}{2}) \alpha^2 \delta_1^2 (1-\delta_2)^2} }_{m_2} 
= \mathcal{O}(\frac{\log(\frac{1}{\beta})}{\varepsilon^2 \alpha^2}).
\end{equation*}

By picking $m \geq \max\{m_1,m_2,m_3\}=\mathcal{O}(\frac{1}{\varepsilon^2 \alpha^2}\log(\frac{1}{\beta}))$, we ensure that
\begin{equation*}
\begin{split}
\Prob\left(| \tilde{d}-d(S) | \ \geq \ \alpha \right)\ & \leq p_1+p_2+p_3 \\
& \leq \delta_3 \beta+\delta_4 \beta+(1-\delta_3-\delta_4) \beta = \beta,
\end{split}
\end{equation*}
which completes the proof.
\end{proof}
\end{toappendix}

Asymptotically, the lower bound on the sample size required to achieve the desired approximation accuracy is not improved.
However, the bound on $p_2$ changes to
\begin{equation*}
p_2 \leq  2 e^{-2m \tanh^2(\frac{\varepsilon}{2}) \alpha^2 \delta_1^2 (1-\delta_2)^2} 
\Leftrightarrow 
m \geq \frac{ \log(\frac{2}{\beta \delta_4}) }{2 \tanh^2(\frac{\varepsilon}{2}) \alpha^2 \delta_1^2 (1-\delta_2)^2},
\end{equation*}
which, as we show, yields an improvement on the actual sample size computed by numerically solving the optimization problem described in Section \ref{subSEC_sample_size}, with the updated cost function $m(\delta_1,\delta_2,\delta_3,\delta_4)$.
The resulting proposed sample size $m$ is also illustrated in Figure \ref{FIG_sample_size_bound} as a function of the privacy budget $\varepsilon$. The proposed sample size is half an order of magnitude less than the sample size required by Algorithm \ref{DworkDE}.

\subsection{Using Continuous Distributions}
As a final remark, we note that an alternative modification to the state of Algorithm \ref{DworkDE} would be to
replace the bitarray $\mathbf{b}$ (which stores bits drawn from either of the two Bernoulli distributions $f_{\text{init}}$ and $f_{\text{upd}}$)
by an array of real numbers $\textbf{x}$, drawn from two continuous distributions.
Our motivation is that the algorithm's output is itself a real number, so by storing a flexible, real value per user (instead of a hard, binary value), we may boost accuracy.
Although the algorithm we derive fails to match the performance of Algorithm \ref{OptBernDE}, it also manages to outperform the original one and provides useful theoretical insights.
We refer the interested reader to \cite{digalakis2018thesis}.

\section{Proposed Pan-Private Density Estimator}
In this section, we propose an additional modification to Algorithm \ref{DworkDE}.
In particular, we extend the static sampling step performed by the original algorithm
to a novel, pan-private version of an adaptive sampling technique,
known as Distinct Sampling \cite{gibbons2001distinct},
which is specially tailored to the distinct count problem.
The novel Algorithm \ref{PPDS} that we derive, which we call \verb|PPDS| (Pan-Private Distinct Sampling), further improves the original one,
especially when the available sample size $m$ is small.

We remark that the modification to the sampling step
is independent of the modifications to the state of the algorithm that we have so far examined.
Hence, the new sampling technique can be combined with any of the density estimators that we have presented.

\subsection{Distinct Sampling: An Overview}
The original Distinct Sampling Algorithm, introduced by Gibbons \cite{gibbons2001distinct}, addresses general distinct count queries over data streams, that is, 
queries that estimate the number of distinct users that satisfy additional query predicates over a set of attributes other than their user ids.
Nevertheless, since our focus is simply on estimating the number of distinct users that appear in the stream (or, equivalently, the stream density),
we use a simplified version of the algorithm.

We proceed with an informal and intuitive presentation of Distinct Sampling.
Recall that, in the \emph{static sampling} approach, the density of $S$ is estimated on a random subset $\mathcal{M} \subseteq \mathcal{U}$ of users whose size $m=|\mathcal{M}|$ is fixed and computed in advance, so as to provide the desired accuracy guarantees.
Notice that $\mathcal{M}$ is sampled before the stream processing phase starts,
so we are able to maintain a size-$m$ bitarray $\mathbf{b}$, with one bit per user in $\mathcal{M}$, that determines whether this particular user has appeared in the stream.
In \emph{Distinct Sampling}, instead of a bitarray, we maintain a sample of user ids, which we call $\mathcal{M}_{\text{DS}}$.
The size $m$ of $\mathcal{M}_{\text{DS}}$ is again fixed (and can be computed in the same manner),
but now $\mathcal{M}$ is sampled \emph{adaptively}, depending on the number of distinct users that have appeared.
In particular, we work as follows.
\begin{itemize}
\item[-] Initially, we set $\mathcal{M}=\mathcal{U}$, that is,
all users are taken into account in the density estimation
and $\mathcal{M}_{\text{DS}}$ is empty.
\item[-] When the stream processing phase begins,
we insert into $\mathcal{M}_{\text{DS}}$ any user that appears for the first time and, as a result, $\mathcal{M}_{\text{DS}}$ contains all the distinct users that have appeared.
\item[-] Assuming that $\mathcal{M}_{\text{DS}}$ did not get full, which means that less than $m$ distinct users appeared, the density is computed by diving the size of $\mathcal{M}_{\text{DS}}$ by the size of $\mathcal{U}$.
\item[-] If $\mathcal{M}_{\text{DS}}$ gets full, we sample uniformly a random subset $\mathcal{M}' \subseteq \mathcal{M}$, that consists of half the users in $\mathcal{M}$.
We evict from $\mathcal{M}_{\text{DS}}$ all the users that were not selected in $\mathcal{M}'$ and,
whenever a new user appears in the stream for the first time, we insert it in $\mathcal{M}_{\text{DS}}$ only if it is also in $\mathcal{M}'$.
Whenever $\mathcal{M}_{\text{DS}}$ refills, we further shrink $\mathcal{M}'$ (randomly) to half its prior size.
\item[-] In the end, we take into account in the density estimation only the users in $\mathcal{M}'$;
if $\mathcal{M}_{\text{DS}}$ got full one time, we consider $\frac{|\mathcal{U}|}{2}$ users, if it got full two times, we consider $\frac{|\mathcal{U}|}{4}$ users, and so forth.
Therefore, in estimating the final density, we again divide the size of $\mathcal{M}_{\text{DS}}$ by the size of $\mathcal{U}$
and then properly scale the result by multiplying it with the inverse of the fraction of $\mathcal{U}$ that we ended up considering.
\end{itemize}

\paragraph*{Technical Details:}
Although a much more detailed (and technical) presentation of Distinct Sampling can be found in \cite{gibbons2001distinct},
there is a key question that we need to answer here as well:
how do we determine efficiently which users are in the set $\mathcal{M}$ over the course of the algorithm?
Explicitly storing the ids of the users in $\mathcal{M}$ requires $\mathcal{O}(|\mathcal{U}|)$ space, which is highly impractical.
The solution proposed is inspired by the hashing technique of the well-known FM-sketch \cite{flajolet1985probabilistic}.

In the next paragraphs, our presentation follows that of Gibbons \cite{gibbons2001distinct}.
For simplicity, we assume that $|\mathcal{U}|=2^Q$. 
There is a level $L$ ($0 \leq L \leq Q$) associated with the procedure, 
that is initially $0$ but is incremented each time the sample size bound $m$ is reached. 
Each user $u \in \mathcal{U}$ is mapped to a random level $\ell=\text{hash}(u)$ ($0 \leq \ell \leq Q$) using an easily computed hash function,
so that each time $u$ appears in the stream, it maps to the same level.
By the properties of the hash function, which we describe later on, a user maps to the $i^{\text{th}}$ level with probability $2^{-(i+1)}$;
thus, we expect about $\frac{|\mathcal{U}|}{2}$ users to map to level $0$, about $\frac{|\mathcal{U}|}{4}$ users to map to level $1$, and so forth.
At any time, we only retain in $\mathcal{M}_{\text{DS}}$ users that map to a level at least as large as the current level $L$
and, therefore, a $2^{-L}$ fraction of $\mathcal{U}$ qualifies to enter $\mathcal{M}_{\text{DS}}$ and is taken into account in the density estimation.
Since each user's level is chosen at random,
$\mathcal{M}_{\text{DS}}$ contains a uniform sample of the distinct users in $S$.

We now briefly describe the hash function $\text{hash: } \mathcal{U} \rightarrow \{ 0,1,...,Q \} $, introduced by Flajolet and Martin \cite{flajolet1985probabilistic} and shown to satisfy, independently for each user $u$, $\Prob ( \text{hash}(u) = i ) = 2^{-(i+1)}$, $\forall i \in \{ 0,...,Q \}$.
The hash function works in two stages.
\begin{itemize}
\item[-] First, each user is mapped uniformly at random to an integer in $\{ 0,...,2^Q-1 \}$.
This is achieved by using a linear hash function $h$, that maps a user id $u$ to $h(u) = (\alpha_{\text{h}} u + \beta_{\text{h}} ) \mod 2^Q $,
where $\alpha_{\text{h}}$ is chosen uniformly at random from $\{ 1,...,2^Q-1 \}$ and $\beta_{\text{h}}$ is chosen uniformly at random from $\{ 0,...,2^Q-1 \}$.
By constraining $\alpha_{\text{h}}$ to be odd, we guarantee that each user maps to a unique integer in $\{ 0,...,2^Q-1 \}$.
\item[-] Second, we define the function $\text{trailing\_zeros}(x)$, that inputs an integer $x$ and outputs the number of trailing zeros in the binary representation of $x$; e.g., $\text{trailing\_zeros}(12) = 2 $ (since the binary representation of $12$ is $1100$).
Then, for any user $u$, $\text{hash}(u) = \text{trailing\_zeros}( h(u) )$.
\end{itemize}

\subsection{Estimator \& Analysis}
In developing a pan-private version of Distinct Sampling, we propose the modifications illustrated in Algorithm \ref{PPDS}.
In particular, we again utilize two Bernoulli distributions, $f_{\text{init}}$ and $f_{\text{upd}}$, selected according to Section \ref{sec::bernoulli-optimal-tuning}.
The algorithm operates in two phases and,
during both phases, the Distinct Sampling procedure (increasing the level and evicting users whenever the sample size bound is reached) is faithfully followed.
During \emph{initialization}, we scan through all users, and randomly insert each of them to $\mathcal{M}_{\text{DS}}$ with probability $p_{\text{init}}$.
To avoid scanning the entire universe, we could select a random subset of $p_{\text{init}} |\mathcal{U}|$ users, compute the resulting level after -supposingly- inserting them to $\mathcal{M}_{\text{DS}}$, and add (all at once) to $\mathcal{M}_{\text{DS}}$ the ones who qualify (based on their levels).
Nevertheless, the scan will be operated only once and offline, so, in most applications, it should not be a problem.
During \emph{stream processing}, and whenever a user $u$ arrives,
if $u$ is already in $\mathcal{M}_{\text{DS}}$, we remove it from $\mathcal{M}_{\text{DS}}$ with probability $1-p_{\text{upd}}$,
whereas, if $u$ is not in $\mathcal{M}_{\text{DS}}$, we add it to $\mathcal{M}_{\text{DS}}$ with probability $p_{\text{upd}}$.

\begin{algorithm}

\caption{\texttt{PPDS}\label{PPDS}}

\DontPrintSemicolon
\SetNoFillComment

\KwIn{Data stream $S$, Privacy budget $\varepsilon$, Sample size $m$}

\KwOut{Density $\tilde{d}(S)$}

Find $Q$ such that $|\mathcal{U}| \leq 2^Q$ \;

Generate hash function $\text{hash}: \mathcal{U} \rightarrow \{ 0,...,Q \} $

Set $p_{\text{init}} = \frac{1}{2}(1-\tanh(\frac{\varepsilon}{2})))$ and $p_{\text{upd}} = \frac{1}{2}(1+\tanh(\frac{\varepsilon}{2}))$\;

\tcc{Initialization Phase}

Initialize $\mathcal{M}_{\text{DS}} = \emptyset$, $L=0$ \;

\For{each $u \in \mathcal{U}$}{

	\If{$\text{hash}(u) \geq L$}{
	
		Add $u$ to $\mathcal{M}_{\text{DS}}$ w/ prob. $p_{\text{init}}$ \;
	
	}
	\If{$|\mathcal{M}_{\text{DS}}| \geq m$}{
	
		\For{each $v \in \mathcal{M}_{\text{DS}}$}{
		
			\If{$\text{hash}(v) \leq L$}{
			
				Remove $v$ from $\mathcal{M}_{\text{DS}}$ \;			
			}
			
		}
		
		Set $L = L + 1$ \;
	
	}

}

\tcc{Stream Processing Phase}

\For{t = 1 \KwTo T}{

	\If{$s_t \in \mathcal{M}_{\text{DS}}$}{
	
		Remove $s_t$ from $\mathcal{M}_{\text{DS}}$ w/ prob. $1-p_{\text{upd}}$ \;
	
	}
	\uElse{
		
		\If{$\text{hash}(s_t) \geq L$}{
	
			Add $s_t$ to $\mathcal{M}_{\text{DS}}$ w/ prob. $p_{\text{upd}}$ \;
	
		}
		\If{$|\mathcal{M}_{\text{DS}}| \geq m$}{
	
			\For{each $v \in \mathcal{M}_{\text{DS}}$}{
		
				\If{$\text{hash}(v) \leq L$}{
				
					Remove $v$ from $\mathcal{M}_{\text{DS}}$ \;			
			
				}
			
			}
		
			Set $L = L + 1$ \;
	
		}
		
	}	
	
}

Return $\tilde{d}(S)\ = \ \frac{1}{\tanh(\frac{\varepsilon}{2})}(\frac{2^L|\mathcal{M}_{\text{DS}}|}{|\mathcal{U}|}-\frac{1}{2} + \frac{1}{2} \tanh(\frac{\varepsilon}{2}) ) \ + \ \text{Laplace}(0,\frac{2^L}{\varepsilon |\mathcal{U}|})$\;

\end{algorithm}

The privacy and accuracy analysis of Algorithm \ref{PPDS} follows the lines of our analysis of Algorithm \ref{DworkDE}.
We first argue about privacy.

\begin{thm}\label{thm_PPDS_priv}
If $\varepsilon \leq \frac{1}{2}$, then Algorithm \ref{PPDS} satisfies $2\varepsilon$-pan-privacy and utilizes all the allocated privacy budget.
\end{thm}

\begin{toappendix}
\begin{proof}[Proof of Theorem~\ref{thm_PPDS_priv}]
The proof is similar to that of Theorem \ref{thm_OptBernDE_priv}.
First, we identify the information that the state of Algorithm \ref{PPDS} consists of.
\begin{itemize}
\item[-] The hash function $\text{hash}()$, 
which (by itself) does not leak information about any user.
\item[-] The sample $\mathcal{M}_{\text{DS}}$, 
which includes the ids of users that \emph{might have appeared} 
and hence must be maintained in a differentially private manner.
\item[-] The current level $L$, 
which is determined by the number of users that are present in $\mathcal{M}_{\text{DS}}$.
Ensuring that a user's presence in $\mathcal{M}_{\text{DS}}$ does not violate its privacy (according to the differential privacy definition) also ensures that the user's privacy is protected against an adversary that observes $L$.
\end{itemize}

We thus have to prove that differential privacy is satisfied against an adversary that observes $\mathcal{M}_{\text{DS}}$.
For any (fixed) $L \in \{ 0,...,Q \}$, only $\frac{|\mathcal{U}|}{2^{L}}$ users \emph{qualify} to be added to $\mathcal{M}_{\text{DS}}$,
so perfect privacy is guaranteed for the remaining users.
Let us now fix a user $u$ who does qualify, that is,
when the adversary gets to see the internal state of the algorithm, $\text{hash}(u) \geq L$. 
Assume that $u$ appears in stream $S$ and does not appear in stream $S'$.
Then,
\begin{equation*}
\begin{split}
& \frac{\Prob\left(u \in \mathcal{M}_{\text{DS}}(S)\right)}{\Prob\left(u \in \mathcal{M}_{\text{DS}}(S')\right)} 
=
\frac{p_{\text{upd}}}{p_{\text{init}}}
=
\ e^{\varepsilon}, \\
& \frac{\Prob\left(u \not\in \mathcal{M}_{\text{DS}}(S)\right)}{\Prob\left(u \not\in \mathcal{M}_{\text{DS}}(S')\right)} 
=
\frac{1-p_{\text{upd}}}{1-p_{\text{init}}}
= 
\ e^{-\varepsilon},
\end{split}
\end{equation*}
so user $u$ is guaranteed $\varepsilon$-differential privacy against an adversary that observes $\mathcal{M}_{\text{DS}}$.
\end{proof}
\end{toappendix}

To argue about the accuracy guarantees of Algorithm \ref{PPDS},
the key thing to notice is that the level $L \in \{0,...,Q \}$ is determined by the particular input stream $S$ (realization of the input).
Therefore, for fixed $S$, $L$ is also fixed (assume $L=\ell$)
and, as a result, the size of the sample $\mathcal{M}$ can be viewed as constant and equal to $\frac{|\mathcal{U}|}{2^{\ell}}$.
The following theorem follows from an analysis identical to the one we performed in the previous section (by simply setting $m=\frac{|\mathcal{U}|}{2^{\ell}}$).

\begin{thm}\label{thm_PPDS_acc1}
For fixed $L=\ell$ and, therefore, for a fixed sample $\mathcal{M}$ that consists of $\frac{|\mathcal{U}|}{2^{\ell}}$ users, Algorithm \ref{PPDS} provides an unbiased estimate $\tilde{d}$ of the density of $S_{\mathcal{M}}$ and has mean squared error
$$\E [ \ (\tilde{d}-d(S_{\mathcal{M}}))^2 \ ] \ \leq \ \frac{2^{\ell-2}}{|\mathcal{U}| \tanh^2(\frac{\varepsilon}{2})} + \frac{2^{2\ell+1}}{|\mathcal{U}|^2 \varepsilon^2}.$$
\end{thm}

\section{Simulation Results}
In this section, we experimentally compare the algorithms we presented.
We conduct all our experiments in MATLAB.
In each experiment, we generate a stream of length $T=10^5$.
The universe is the set $\mathcal{U}=\{1,...,10^5\}$ and the stream is either uniform or zipfian (with parameter 1).
In the former case we expect to observe a stream density of $0.63$, while in the latter of $0.25$.
In what follows, we use the abbreviation
\verb|Dwork| to refer to the Conventional Pan-Private Density Estimator, that is, Algorithm \ref{DworkDE},
\verb|OptBern| to refer to the Improved Pan-Private Density Estimator, that is, the Algorithm \ref{OptBernDE}, and
\verb|PPDS| to refer to the Proposed Pan-Private Density Estimator, that is, the Algorithm \ref{PPDS}.

In our first set of experiments, we compare the first two estimators, \verb|Dwork| and \verb|OptBern|.
The evaluation metric we use is the probability $\Prob(| \tilde{d}-d(\mathcal{S}) | \geq \alpha)$, which we call \enquote{probability of error} for simplicity.
Our experiments are for fixed sample size (as a fraction of the universe size), for fixed $\alpha$ and for varying privacy budget $\varepsilon$.
Therefore, we do not take into account the input parameter $\beta$ (desired upper bound on the probability of error) to pick the proper sample size;
instead, we fix the sample size and examine how each algorithm performs in terms of the probability of error it actually achieves.
The illustrated probability of error is the empirical probability, computed over $1000$ repetitions per $\varepsilon$. 
The experimental results, which we present in Figure \ref{FIG_prob_error},
confirm the superiority of algorithm \verb|OptBern|,
regardless of the distribution of the input stream.
Another key thing to notice is that the bounds on the probability of error which we computed theoretically are not tight;
although we did compute the tightest version of the bound,
the resulting probability was significantly larger than all the empirical probabilities we plot in Figure \ref{FIG_prob_error}.

\begin{figure}
    \centering
    \includegraphics[width=\columnwidth]{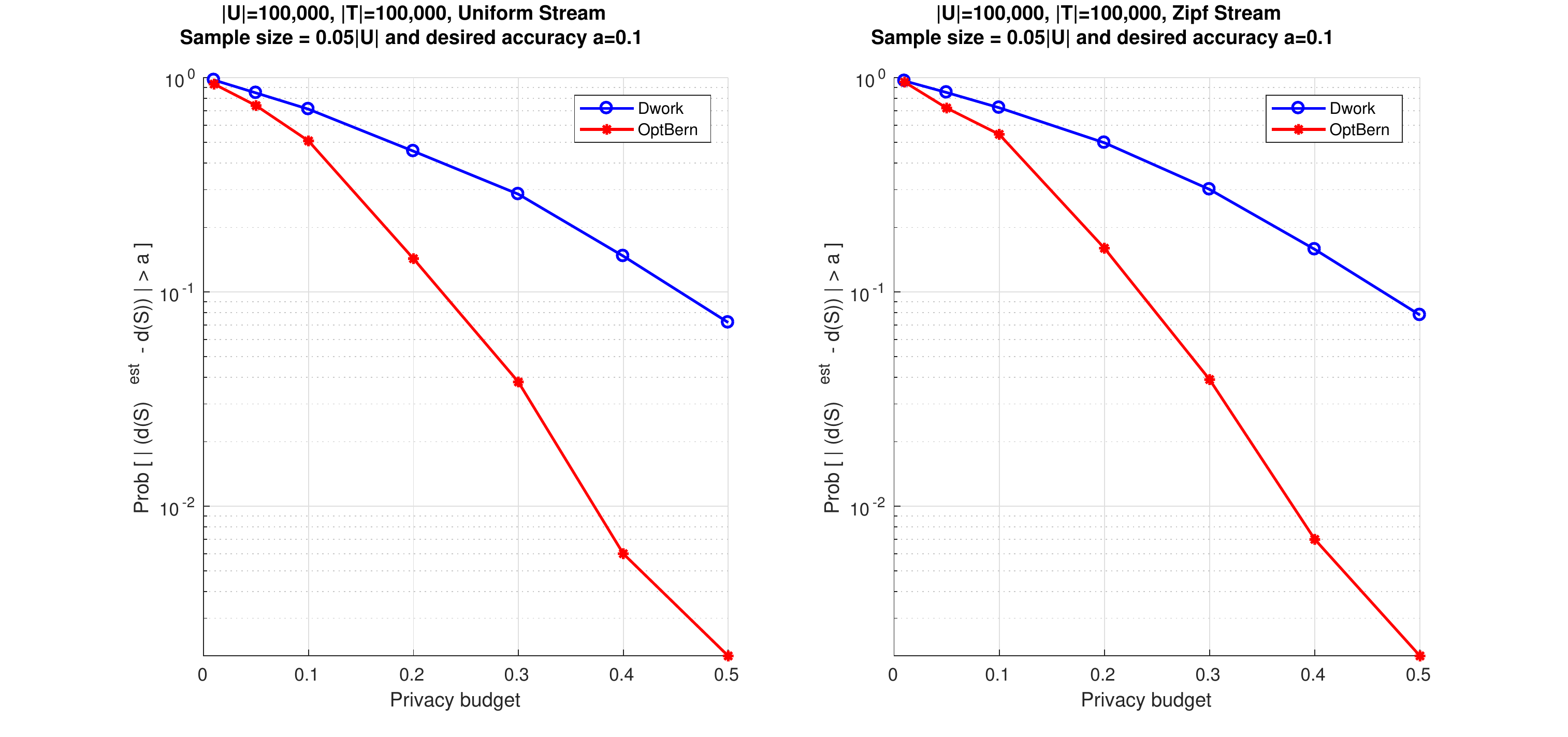}
    \caption{Empirical probability of error as a function of the privacy budget. We set $|\mathcal{U}|=100,000$, $T=100,000$, sample size $m = 0.05|\mathcal{U}|$, and desired accuracy $\alpha=0.1$.}
	\label{FIG_prob_error}
\end{figure}

Next, we compare all three estimators and examine our second evaluation metric, that is, the mean squared error (MSE) of the algorithms, as a function of the allocated privacy budget.
We remark that we again do not take into account the input parameters $\alpha,\beta$ in order to pick the proper sample size,
which we fix as a constant fraction of the universe size, and examine how each algorithm performs for varying $\varepsilon$.
For each $\varepsilon$, we independently repeat the experiment $300$ times.
In Figure \ref{FIG_MSE_theor}, we only plot the experimental MSE of algorithms \verb|Dwork| and \verb|OptBern| and compare it with their theoretical MSE.
The experimental and theoretical MSE coincide for both algorithms.
In Figure \ref{FIG_MSE}, we plot the experimental MSE of all three algorithms.
We observe that both \verb|PPDS| and \verb|OptBern| significantly outperform \verb|Dwork| and that the performance of all algorithms is robust to the input stream's distribution; the differences between uniform and zipfian are insignificant.
\verb|PPDS| seems to perform noticeably better than \verb|OptBern| when the stream is sparser (i.e., Zipf Stream) and when the allocated sample size is small compared to the universe size.

\begin{figure}
\centering
\begin{subfigure}{\columnwidth}
  \centering
  \includegraphics[width=\columnwidth]{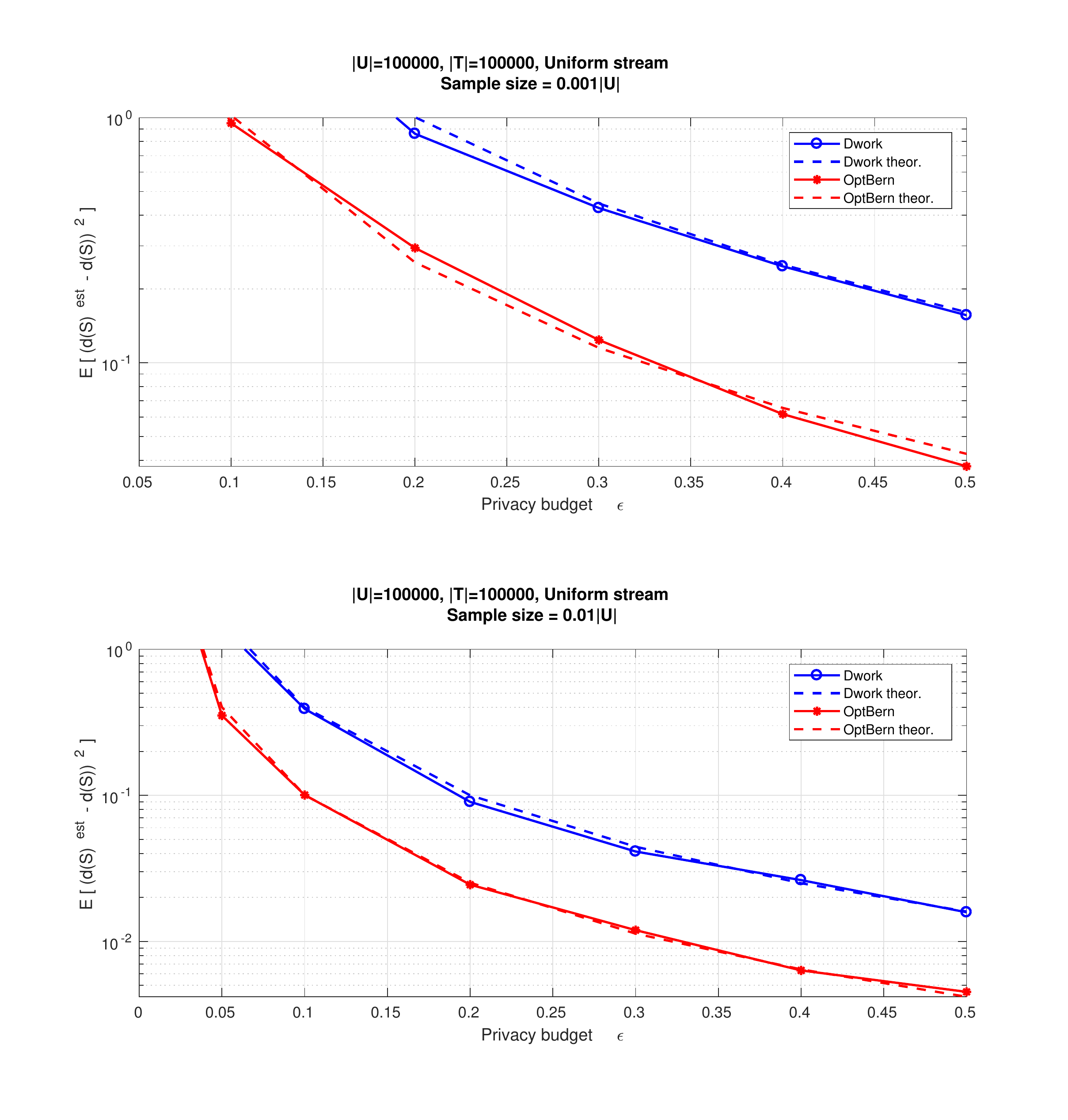}
  \caption{Uniform stream.}
\end{subfigure}\\
\begin{subfigure}{\columnwidth}
  \centering
  \includegraphics[width=\columnwidth]{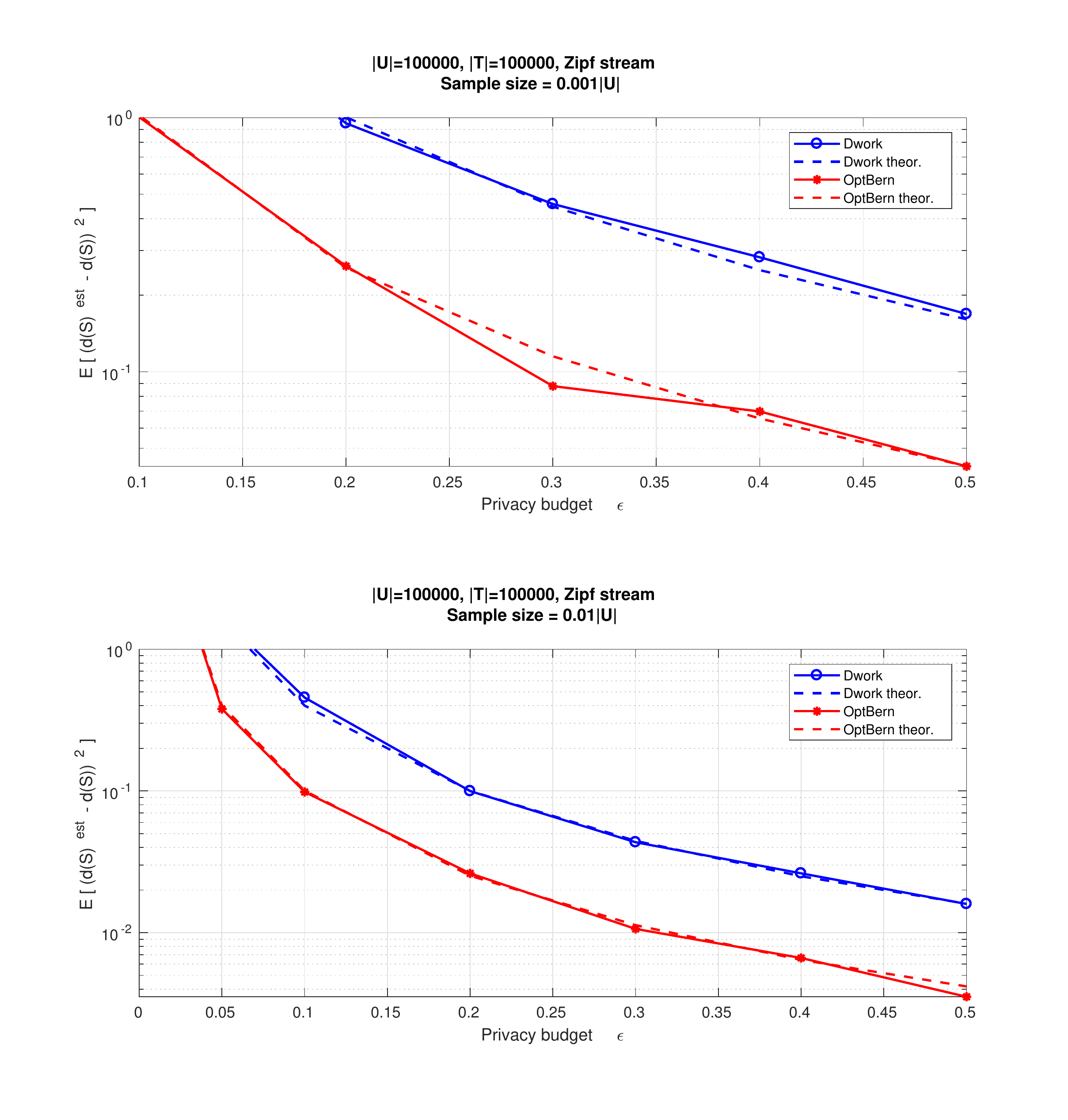}
  \caption{Zipf stream.}
\end{subfigure}
\caption{Theoretical and empirical MSE as a function of the privacy budget and for various sample sizes ($m = 0.001|\mathcal{U}|$ above, $m = 0.01|\mathcal{U}|$ below). We set $|\mathcal{U}|=100,000$, $T=100,000$.}
\label{FIG_MSE_theor}
\end{figure}

\begin{figure}
\centering
\begin{subfigure}{\columnwidth}
  \centering
  \includegraphics[width=\columnwidth]{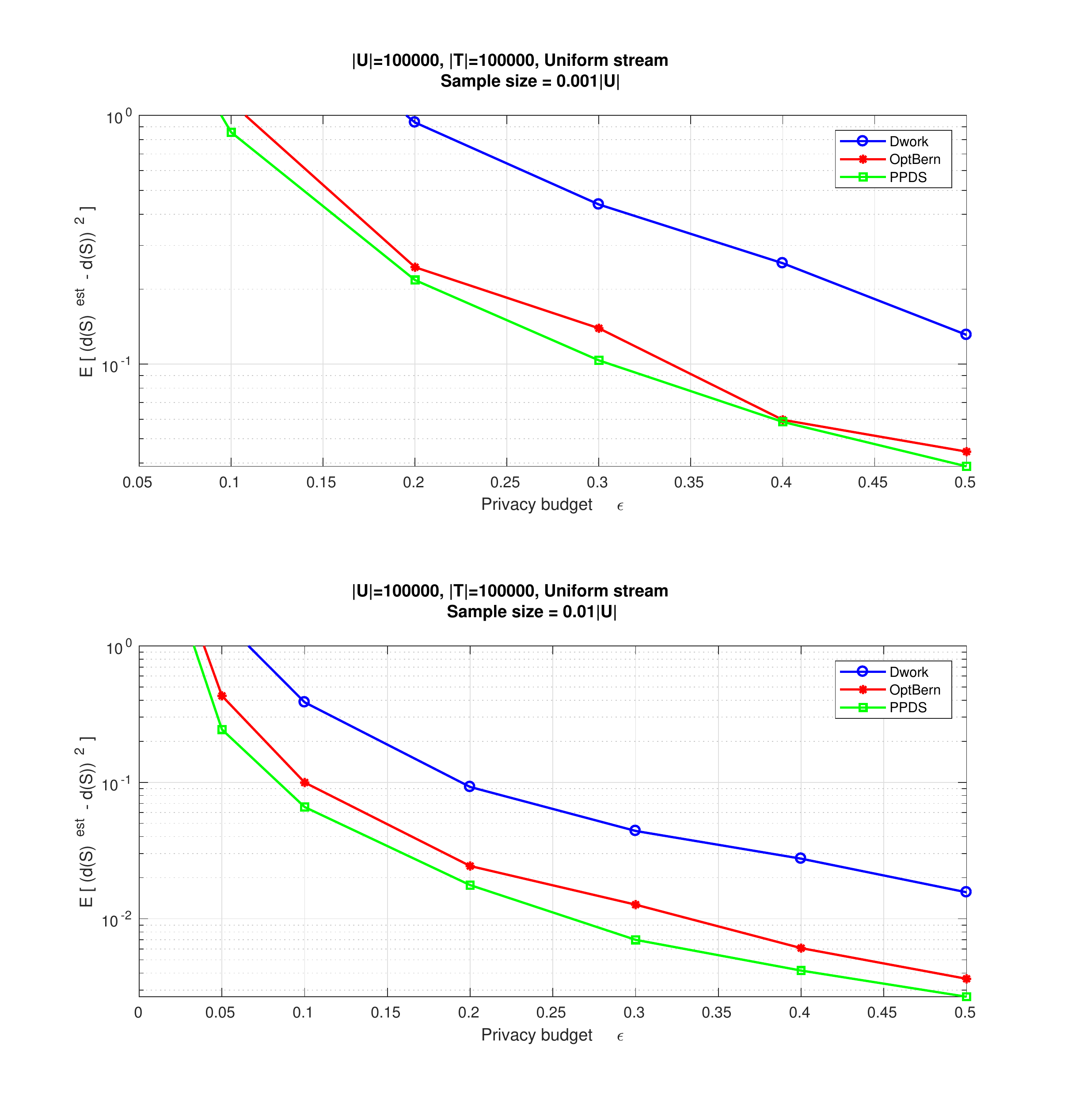}
  \caption{Uniform stream.}
\end{subfigure}\\
\begin{subfigure}{\columnwidth}
  \centering
  \includegraphics[width=\columnwidth]{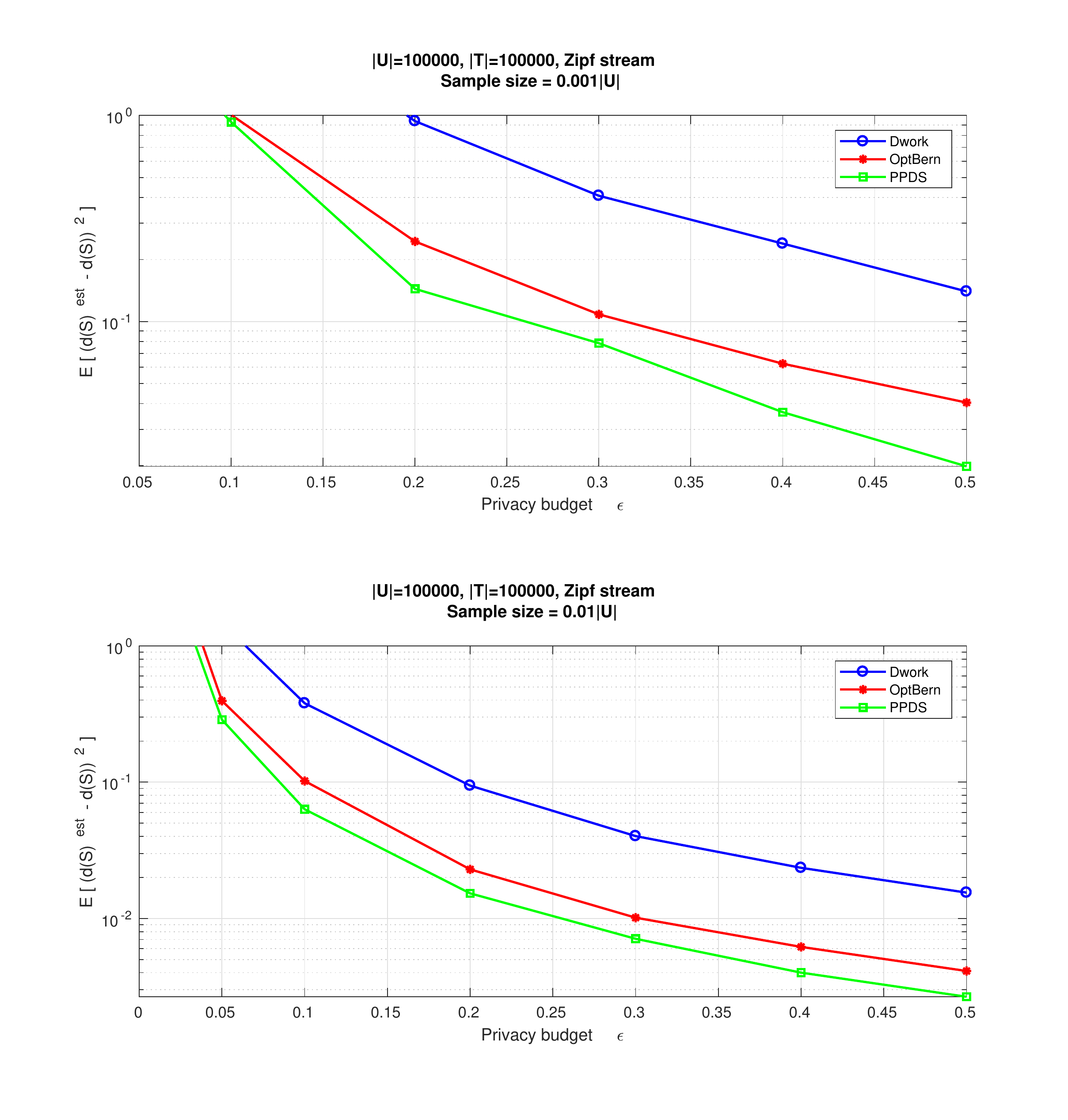}
  \caption{Zipf stream.}
\end{subfigure}
\caption{Empirical MSE as a function of the privacy budget and for various sample sizes ($m = 0.001|\mathcal{U}|$ above, $m = 0.01|\mathcal{U}|$ below). We set $|\mathcal{U}|=100,000$, $T=100,000$.}
\label{FIG_MSE}
\end{figure}

Finally, we examine the MSE of all three algorithms for varying sample size $m$.
We again ignore the input parameters $\alpha,\beta$ and fix the privacy budget $\varepsilon=0.2$.
We express the sample size as a fraction of the universe size, which we call sample percentage, 
and, for each sample percentage, we independently repeat the experiment $300$ times.
Figure \ref{FIG_MSE_sample_size} illustrates not only the superiority of our algorithms over \verb|Dwork|, but also the effect of the sample size on algorithms \verb|OptBern| and \verb|PPDS|.
For very small sample percentages (i.e., $\frac{1}{1000}$), both algorithms struggle,
but still significantly outperform \verb|Dwork|.
Above a certain threshold, the superiority of \verb|PPDS| is clear and the difference between the algorithms' MSE is maximized.
As the sample size increases, the performance of \verb|OptBern| improves and converges to that of \verb|PPDS|.
When the sample percentage is set to one, that is, no sampling is performed, the two algorithms are identical.

\begin{figure}
\centering
\begin{subfigure}{.5\columnwidth}
  \centering
  \includegraphics[width=\columnwidth]{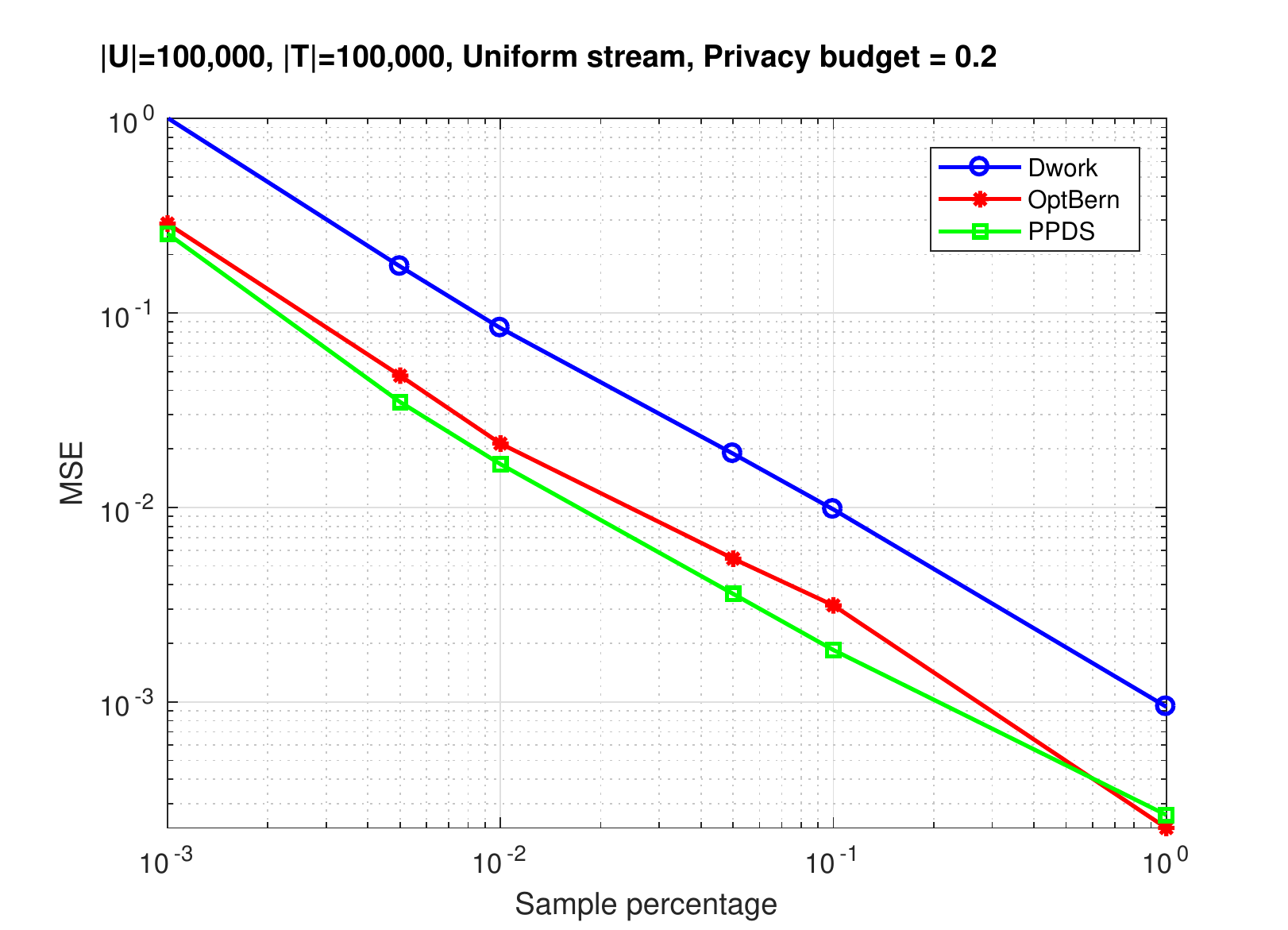}
  \caption{Uniform stream.}
\end{subfigure}%
\begin{subfigure}{.5\columnwidth}
  \centering
  \includegraphics[width=\columnwidth]{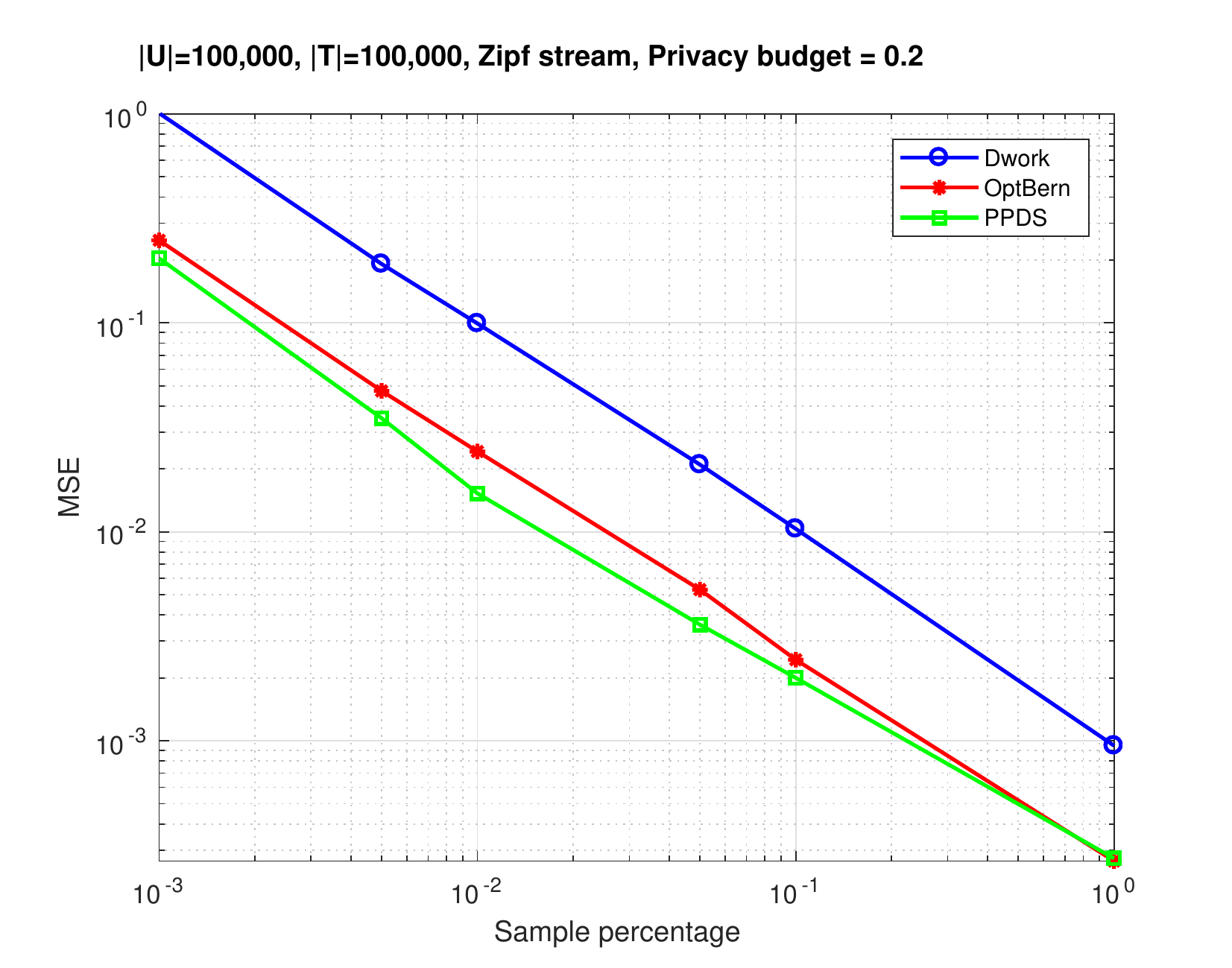}
  \caption{Zipf stream.}
\end{subfigure}
\caption{Empirical MSE as a function of the sample size. We set $|\mathcal{U}|=100,000$, $T=100,000$, $\varepsilon=0.2$.}
\label{FIG_MSE_sample_size}
\end{figure}

\section{Conclusions}
In this work, we addressed the problem of pan-private stream density estimation.
We analyzed for the first time the sampling-based pan-private density estimator proposed by Dwork et al. \cite{dwork2010pan},
and identified that it does not use all the allocated privacy budget.
We managed to outperform the original algorithm both theoretically and experimentally by proposing novel modifications
that are based on optimally tuning the Bernoulli distributions it uses
and on reconsidering the sampling step it performs.

%
%


\bibliographystyle{ACM-Reference-Format}
\bibliography{panprivacy-ref}

\end{document}